%% file: main_full_version.tex
\documentclass[runningheads]{llncs}
\usepackage[T1]{fontenc}
\usepackage{graphicx}
\usepackage{booktabs} 
\usepackage[ruled]{algorithm2e} 
\usepackage{amsmath}
\usepackage{thmtools}
\usepackage{bbm}
\usepackage{cleveref}

\usepackage[libertine]{newtxmath}

\usepackage{amssymb}
\usepackage{float} 
\usepackage{algorithmic}
\usepackage{comment}

\usepackage{caption}
\usepackage{subcaption}
\usepackage{tikz}
\usepackage{ulem}
\usepackage{empheq}
\usepackage{fancybox}
\usepackage{pgfplots}

\renewcommand{\ts}{\theta^-}
\newcommand{\tf}{\theta^+}
\newcommand{\OPT}{\texttt{OPT}\xspace}
\newcommand{\SLEEPY}{\texttt{SLEEPY}\xspace}
\newcommand{\LPT}{\texttt{LPT}\xspace}
\newcommand{\GSLEEPY}{\texttt{Generalized SLEEPY}\xspace}

\begin{document}
\title{Online Makespan Minimization: Beat LPT by Dynamic Locking \thanks{This work was completed when Zhaozi Wang was an undergraduate student at Shanghai Jiao Tong University. He is now a Ph.D. candidate at New York University.}}


\author{Zhaozi Wang\inst{1}\orcidID{0009-0002-0177-9426} \and
Zhiwei Ying\inst{1}\orcidID{0009-0008-3200-9575} \and
Yuhao Zhang\inst{1}\orcidID{0000-0001-9330-1926}}
\authorrunning{Z. Wang et al.}
%

\institute{Shanghai Jiao Tong University, No. 800 Dongchuan Road, Minhang District, Shanghai, China‌\\
\email{\{ocmykr2,yingzhiwei,zhang\_yuhao\}@sjtu.edu.cn}}

\maketitle

\input{abstract-wine}
\input{introduction}
\input{related_works}
\input{preliminary}

\input{overview}

\input{analysis_criticaljob}

\input{analysis_eff_pre}

\input{analysis_eff_c1}
\input{analysis_eff_c2_lipics}

\input{bound}

\input{alphagamma_m=3}
\input{hard_instance}
\input{conclusion}


\bibliographystyle{wine/splncs04}
\bibliography{sample}

\end{document}

%% file: abstract-wine.tex
\begin{abstract}
Online makespan minimization is a fundamental problem in scheduling. In this paper, we investigate its over-time formulation, where each job has a release time and a processing time. A job becomes known only at its release time and must be scheduled on a machine thereafter. The Longest Processing Time First (LPT) algorithm, established by Chen and Vestjens (1997), achieves a competitive ratio of $1.5$. For the special case of two machines, Noga and Seiden introduced the SLEEPY algorithm, which achieves a tight competitive ratio of $1.382$. However, for $m \geq 3$, no known algorithm has convincingly surpassed the long-standing $1.5$ barrier.

We propose a natural generalization of SLEEPY and show this simple approach can beat the $1.5$ barrier and achieve $1.482$-competitive when $m=3$. However, when $m$ becomes large, we prove this simple generalization fails to beat $1.5$. Meanwhile, we introduce a novel technique called dynamic locking to overcome this new challenge. As a result, we achieve a competitive ratio of $1.5-\frac{1}{O(m^2)}$, which beats the LPT algorithm ($1.5$-competitive) for every constant $m$. 

\keywords{Online scheduling \and Makespan minimization \and Dynamic locking.}

\end{abstract}

%% file: introduction.tex
\section{Introduction}

Makespan Minimization is a well-established objective in the field of scheduling. In the standard case of identical machines, we are presented with $m$ identical machines and $n$ jobs that require scheduling. The objective is to assign each job to a specific machine to minimize the makespan, which corresponds to the last completion time among all machines. Traditional offline algorithm design focuses on accomplishing this task efficiently under the assumption that we know all the jobs' information at the beginning. However, the assumption does not hold in many real-world applications, where we are required to make online decisions. Minimizing makespan in an online scenario has garnered significant attention since the 1960s~\cite {DBLP:journals/siamam/Graham69}. There are two types of online formulations of makespan minimization problems.
\begin{itemize}
    \item Sequential (list) formulation: jobs are all released at $0$ but following an online order that the next job is revealed only after we decide irrevocably which machine is allocated for the current one.
    \item Over-time arrival formulation: Each job $j$ is associated with a release time $r_j$, and is revealed by the algorithm at $r_j$. The algorithm can start processing a job at an idle machine after the job's release time without preemption. 
\end{itemize}

Our paper focuses on the second one, under the standard competitive analysis for online algorithms. An algorithm is said to be $c$-competitive, or achieves a competitive ratio of $c$ if its makespan is at most $c$ times the optimal makespan under all possible inputs. 

Let us compare the two models from the perspective of motivation. The sequential formulation naturally fits scenarios involving short-term load balancing tasks, such as scheduling across multiple CPUs over a brief time window. In these cases, many tasks arrive almost simultaneously, so their release times can be treated as zero. Moreover, such settings typically require immediate dispatch, meaning that jobs must be scheduled in the order they arrive. In this context, minimizing the makespan serves as a reasonable objective, as it promotes balanced load distribution among the CPUs. On the other hand, the over-time formulation is more suitable for long-term manufacturing scenarios. In such settings, tasks arrive gradually over time, as part of a batch that must be completed together to trigger downstream operations. Here, the makespan becomes a natural objective, and release times play a crucial role, as jobs no longer arrive simultaneously. Given this long-term perspective, we no longer require immediate dispatch; instead, we can determine which jobs to process over time without adhering to a fixed decision order.

A common concern with using the makespan objective in the presence of release times is whether it still effectively captures load balancing across machines. In the offline setting, where all release times are known in advance, one might worry that the makespan could be dominated by the latest released jobs, potentially overlooking imbalances in scheduling earlier tasks. However, this concern is naturally mitigated under the worst-case competitive analysis framework for online algorithms. If an online algorithm allocates early jobs in an unbalanced manner, but performs well for later jobs, an adversary could exploit this by releasing only the early jobs, resulting in a poor competitive ratio. Consequently, any online algorithm with a good competitive ratio must inherently maintain balanced scheduling at all times. This makes the makespan objective with release times well-suited for long-term scheduling tasks, particularly when the underlying goal is to balance load across machines.


In 1997, Chen and Vestjens \cite{DBLP:journals/orl/ChenV97} proposed an online LPT (Longest Processing Time First) algorithm described as follows: at each time that there is an idle (not processing any job) machine and several released but not scheduled jobs, schedule the longest job on that machine. They prove that LPT is $1.5$-competitive, and the analysis is tight for $m\geq 2$. The tightness is established by the following one-one-two hard instance: initially, two size-$1$ jobs are released at time $0$, and LPT immediately schedules them on two separate machines. Subsequently, just after the commencement of these two jobs, $m-1$ jobs with a processing time of $2$ are released. No other jobs are released later. LPT finishes these jobs at $3$, while the optimal solution can complete them at $2$ by placing the first two small jobs on the same machine. On the hardness side for general algorithms, they show that no deterministic online algorithm has a competitive ratio better than $(\frac{5-\sqrt{5}}{2})\approx 1.382$ for $m=2$ and $1.347$ for $m\geq 3$, which is further improved to $1.3596$ by Li and Zhang \cite{DBLP:journals/jossac/LiZ16}. The most natural open question is: can we beat this simple LPT strategy in the over-time Online Makespan Minimization problem?

Let us begin by considering how to beat the competitive ratio of $1.5$ in the classic one-one-two hard instance. When two size-1 jobs arrive at time 0, processing both immediately yields a competitive ratio of $1.5$. This naturally suggests a strategy: process one job and delay the other. In 2001, Noga and Seiden~\cite{DBLP:journals/tcs/NogaS01} applied this insight and proposed an algorithm called $\tt{SLEEPY}$ for the 2-machine case. Remarkably, they achieved the optimal competitive ratio of $(\frac{5 - \sqrt{5}}{2}) \approx 1.382$. However, when the number of machines becomes larger, there are more technical challenges. In 2018, Huang et al.~\cite{DBLP:conf/approx/0002KTWZ18} managed to break the $1.5$ barrier for general $m$ by utilizing an additional power called restart (which involves stopping the processing of a job but reprocessing the whole job later). Alongside this, they introduced a key analytical tool, the left-over lemma, which provides a bound comparing the total amount of work completed by the algorithm and by the optimal schedule over a given time interval. Despite these advances, whether LPT and its $1.5$ competitive ratio can be outperformed in the original (non-restart) setting for any $m \geq 3$ has remained an open question for decades, even for small constants such as $m = 3$ or $m = 4$.

We partially answer this long-standing open question on the positive end. We generalize the \SLEEPY algorithm and propose the \GSLEEPY algorithm that works for general $m$, along with a novel technique called dynamic locking. We prove that it is $\left(1.5-\frac{1}{O(m^2)}\right)$-competitive for general $m$. In conclusion, we beat LPT for all the cases where $m$ is a constant. 

\begin{restatable}{theorem}{ratiogeneral}
\label{thm:ratio_m>3}
\GSLEEPY achieves a competitive ratio of $1.5-\frac{1}{O(m^2)}$ for the case $m\geq 4$, with dynamic locking.  
\end{restatable}

Interestingly, in the case $m=3$, a good competitive ratio can be achieved more straightforwardly without the dynamic locking technique, which is $1.482$. The result is formalized in the following theorem.

\begin{theorem}
\label{thm:ratio_m=3}
\GSLEEPY achieves a competitive ratio of $1.482$ for the case $m=3$, without dynamic locking.
\end{theorem}

However, the dynamic locking strategy is crucial when $m$ becomes large. In particular, we show that the competitive ratio becomes $1.5$ without dynamic locking when $m\geq 6$.

\begin{restatable}{theorem}{dynamiclocking}
\label{thm:dynamic locking}
    There is a hard instance for every $m\geq 6$, such that the competitive ratio of \GSLEEPY becomes at least $1.5$ without dynamic locking.
\end{restatable}

The following table summarizes the previous results and our contribution to the online (over-time) machine minimization problem.  
\begin{table}[htbp]
    \centering
    \caption{Competitive ratios (our improvement is written in red). DL means dynamic locking.}
    \begin{tabular}{ |c |c |c | c|c|}
      \hline
      & $m=1$ & $m=2$ & $m=3$ & $m\ge4$\\ 
      \hline
      \texttt{LPT\cite{DBLP:journals/orl/ChenV97}} & $1$ & $1.5$ & $1.5$ & $1.5$ \\  
      \hline
      \texttt{SLEEPY\cite{DBLP:journals/tcs/NogaS01}} & $1$ & $1.382$ (optimal) & undefined & undefined \\
      \hline
      \GSLEEPY (without DL) & $1$ & $1.382$ (optimal) & \textcolor{orange}{1.482} & \textcolor{orange}{$> 1.5\ (m\geq 6)$} \\    
      \hline
      \GSLEEPY (with DL) & $1$ & $1.382$ (optimal) & \textcolor{orange}{1.482} & \textcolor{orange}{$1.5-1/O(m^2)$} \\    
      \hline
    \end{tabular}
\end{table}

\subsection{Our Techniques}

Most analyses of competitive or approximation ratios for makespan minimization rely on two fundamental techniques for upper-bounding the power of the optimal solution (\OPT): the \emph{bin-packing} argument and the \emph{efficiency} argument. The bin-packing argument limits how many large jobs \OPT can schedule; for example, it cannot process more than $m$ jobs of size more than $0.5 \cdot \OPT$, as each machine can accommodate at most one such job. The efficiency argument, on the other hand, bounds the total workload \OPT can complete in a time window; a basic form states that \OPT cannot process more than $tm$ units of work in $[0, t)$. To prove that an online algorithm is $c$-competitive, we typically argue by contradiction: if the algorithm’s makespan exceeds $c \cdot \OPT$, then either there are too many large jobs (violating bin-packing), or too much workload is completed by the algorithm (violating efficiency).

We build on the locking idea of the \SLEEPY algorithm, which is proposed by Noga and Seiden~\cite{DBLP:journals/tcs/NogaS01}, originally designed for $m=2$. This algorithm follows the \LPT strategy but, whenever it starts a job $j$ at time $s_j$, it locks the other machine until $s_j + \alpha p_j$, where $\alpha$ is the locking parameter and $p_j$ is the processing time of job $j$. In other words, no job may start on the other machine before the lock expires.
Consider the classic one-one-two hard instance: \LPT achieves a competitive ratio of $1.5$, as the last large job is delayed from its release time $r_n$ to its start time $s_n$, with this delay amounting to $0.5 \cdot \OPT$ while both machines remain busy. The delayed interval $[r_n, s_n)$ must thus be a busy period on all machines. To create this busy delayed period, the adversary can release two size-1 jobs at time $0$, prompting \LPT to start them just before $r_n = 0 + \epsilon$, whereas \OPT can schedule both on a single machine.
However, under the locking strategy of \SLEEPY, a gap is introduced between the start times of these size-1 jobs. If the adversary still attempts to enforce a busy period at $r_n$, the first job must start strictly before $r_n$. As a result, the size of this first job must exceed the length of $[r_n, s_n)$, which is $0.5 \cdot \OPT$. This prevents \OPT from assigning both jobs to the same machine, thereby breaking the hard instance. We extend this idea to $m \geq 3$ by locking \emph{all} other machines whenever a job starts.

However, locking comes with drawbacks. It can waste machine capacity by leaving machines idle, even when there are released jobs that have not yet started. This causes the algorithm to complete less total work in certain intervals, sometimes even less than \LPT—which weakens efficiency arguments. For $m = 2$, this is manageable, since at most one machine is locked. But for $m \geq 3$, we must formally analyze the impact. We address this by combining our analysis framework with the left-over lemma from Huang et al.~\cite{DBLP:conf/approx/0002KTWZ18}.

A deeper challenge is that the last job may be delayed either because machines are busy or because they are locked. This distinction complicates bin-packing arguments: the delayed interval $[r_n, s_n)$ is no longer purely a busy period—it may contain idle time caused by locking. For instance, if job $n$ starts late because a larger job $n'$ (with $p_{n'} > p_n$ and $r_{n'} = r_n$) is locking machines, then $s_n = s_{n'} + \alpha p_{n'}$. The interval $[r_n, s_{n'})$ may be busy, but $[s_{n'}, s_n)$ is idle. As a result, the effective busy period may be shorter than $0.5 \cdot \OPT$. Even if $m$ jobs start before $r_n$, and we can show gaps between their start times—making them larger than $[r_n, s_{n'})$—it remains unclear whether two such jobs can fit on one machine in \OPT, since the size increase may not offset the busy-period reduction.

For $m = 2$, Noga and Seiden~\cite{DBLP:journals/tcs/NogaS01} showed that such locking-induced delays cannot happen, using an efficiency argument. We extend this insight to $m = 3$. However, the approach fails when $m$ becomes large (see \Cref{thm:dynamic locking}, proved in \Cref{sec:hard}).

To resolve this, we introduce a novel technique called \textbf{dynamic locking}. Our idea is to assign larger locking parameters to jobs that start before $r_n$ than to the job $n'$. In an online setting, this is nontrivial, since we cannot know in advance which job will be the last to start before $r_n$. To handle this, we set the locking parameter $\alpha$ dynamically based on each job’s processing time $p_j$ and start time $s_j$, decreasing $\alpha$ as $s_j/p_j$ increases. This implicitly ensures that the job $n'$ receives a smaller locking parameter than the jobs that start before $r_n$, leading to a larger gap between their start times. This gap prevents \OPT from fitting both jobs on a single machine.

\subsection{Paper Organization}
We first introduce some basic notations and our algorithms in \Cref{sec:pre}. Then, we briefly introduce our analysis framework and some basic properties in \Cref{sec:overview}. Next, we complete the proof of the main theorem (\Cref{thm:ratio_m>3}) in \Cref{sec:basic}, \Cref{sec:critical}, and \Cref{sec:bound}. The case for $m=3$ (proof of \Cref{thm:ratio_m=3}) is discussed in \Cref{sec:m=3}. The hardness result (\Cref{thm:dynamic locking}) is proved in \Cref{sec:hard}. 

%% file: related_works.tex
\subsection{Further Related Work}

Focusing on the over-time formulation of the online makespan minimization problem, Li et al.~\cite{DBLP:journals/ol/LiY16} studied a special case known as the kind release time (KRT) condition, where no jobs are released while all machines are busy. They showed that under this assumption, the classic LPT algorithm achieves a competitive ratio of $\frac{5}{4}$ when $m = 2$, complementing the general upper and lower bounds discussed earlier. In the preemptive setting, jobs' processing may be interrupted and resumed later (without migration). An optimal preemptive algorithm for identical machines was presented by Chen et al.~\cite{DBLP:journals/orl/ChenVW95}, who showed that the competitive ratio is $1 / (1 - (1 - \frac{1}{m})^m)$, which equals $\frac{4}{3}$ when $m = 2$ and converges to $\frac{e}{e - 1}$ as $m \to \infty$. 

Another commonly studied model is the sequential formulation, where all jobs are released at time zero but are revealed one by one in an online manner, requiring the algorithm to assign each job to a machine upon its arrival. This setting corresponds to the classical online load balancing problem. A foundational result was established by Graham~\cite{DBLP:journals/siamam/Graham69}, who showed that the greedy algorithm—assigning each job to the machine where it can start earliest—achieves a competitive ratio of $2 - \frac{1}{m}$, which is optimal for $m \leq 3$. Subsequent improvements for larger values of $m$ have been explored by Bartal et al.~and Karger et al.~\cite{DBLP:journals/jcss/BartalFKV95,DBLP:journals/jal/KargerPT96}. A lower bound of $1.88$ for deterministic algorithms in the non-preemptive setting was established by Rudin et al.~\cite{DBLP:journals/siamcomp/RudinC03}. 
Extensions of this problem to uniform and related machines have been investigated by Aspnes et al.~and Jez et al.~\cite{DBLP:conf/stoc/AspnesAFPW93,DBLP:journals/scheduling/JezSSB13}. The best known upper bound of $5.828$ for this setting is achieved by the algorithm of Berman et al.~\cite{DBLP:journals/jal/BermanCK00}, while a lower bound of $2.564$ was established by Ebenlendr et al.~\cite{DBLP:journals/mst/EbenlendrS15}. Additional results and a broader overview of this setting can be found in the works of Albers et al., Galambos et al., and Dwibedy et al.~\cite{DBLP:journals/siamcomp/Albers99,DBLP:journals/siamcomp/GalambosW93,DBLP:journals/sigact/DwibedyM22}.

In recent years, various extensions of these basic models have been proposed. One such extension includes the use of a reordering buffer, as studied in Englert et al.~\cite{DBLP:journals/siamcomp/EnglertOW14}. Another notable direction is the model of parallel schedules, where the online algorithm is permitted to construct several candidate schedules and select the best one at the end. This approach was analyzed in Albers et al.~\cite{DBLP:journals/algorithmica/AlbersH17}, where a $(\frac{4}{3} + \epsilon)$-competitive algorithm was proposed for any $0 < \epsilon \leq 1$, using $\left(\frac{1}{\epsilon}\right)^{O(\log(1/\epsilon))}$ schedules. Additionally, a $(1 + \epsilon)$-competitive algorithm that uses $\left(\frac{m}{\epsilon}\right)^{O(\log(1/\epsilon)/\epsilon)}$ schedules was also presented.

%% file: preliminary.tex
\section{Preliminaries}
\label{sec:pre} 
First, we formally define the model. We have $m$ identical machines, and $n$ jobs arrive online, where each job $j$ is associated with a release time $r_j$ and a processing time $p_j$. The online algorithm has no information about job $j$ before time $r_j$. At any moment $t$, it can choose to process one of the released jobs $j$ on a machine non-preemptively. We denote the start time and completion time of job $j$ by $s_j$ and $C_j = s_j + p_j$, respectively. The objective is to minimize the makespan, i.e., the maximum completion time among all jobs, $\max_j C_j$.

In our context, $\mathcal{J}$ denotes the set of jobs, while $\mathcal{M}$ denotes the set of machines. The schedule produced by our algorithm, along with its corresponding makespan, is denoted by $\sigma(\mathcal{J})$. Meanwhile, $\pi(\mathcal{J})$ represents the optimal offline schedule for the given job set $\mathcal{J}$, along with its makespan. When $\mathcal{J}$ is clear from context, we may simplify the notation to $\sigma$ and $\pi$. We focus on the standard competitive analysis in our paper. We say that an algorithm is $(1+\gamma)$-competitive if $\frac{\sigma(\mathcal{J})}{\pi(\mathcal{J})} \leq 1+\gamma$ holds for all possible instances $\mathcal{J}$.


\subsection{The \GSLEEPY Algorithm} 
Our algorithm is an LPT-style algorithm embedded with a specific machine-locking strategy: whenever a job starts processing on a machine, all other machines are locked for a certain period. Before presenting our algorithm in detail, we first define the state of jobs and machines at a given time $t$.
\[
\text{A job is} 
\begin{cases}
    \textbf{released} & \text{if } t \geq r_j;\\
    \textbf{processing} & \text{if } s_j \le t < C_j = s_j + p_j; \\
    \textbf{finished} & \text{if } t \geq C_j ; \\
    \textbf{pending} & \text{if } r_j \le t < s_j.
\end{cases}
\]

\[
\text{A machine is} 
\begin{cases}
    \textbf{busy} & \text{if it is currently processing a job $j$;} \\
    \textbf{locked} & \text{if the algorithm does not allow it to start a new job;} \\
    \textbf{idle} & \text{if it is neither busy nor locked.}\\
\end{cases}
\]

We now present the algorithm. Let $\lambda \ge 1$ and $\alpha \in [0, \gamma]$ be the locking parameters. The \GSLEEPY algorithm with dynamic locking operates as follows. At any moment $t$, if there is at least one idle machine, the algorithm selects an arbitrary idle machine and performs the following two steps:

\begin{enumerate}
\item \textbf{LPT Strategy:} Assign the longest pending job $j$ to the selected idle machine and start processing it.
\item \textbf{Dynamic Locking Strategy:} Lock all machines until time $s_j + \alpha_j p_j$, where the locking parameter is given by $\alpha_j = \lambda^{-\frac{s_j}{p_j}} \alpha \le \alpha$. Note that if another machine is already busy, we can still treat it as locked until $s_j + \alpha_j p_j$; a machine can thus be marked as both busy and locked simultaneously.
\end{enumerate}

Note that when we refer to the \GSLEEPY algorithm without dynamic locking, we mean that the locking parameter $\alpha_j$ is set to a fixed constant $\alpha$ (i.e., $\lambda = 1$). We show that even this simpler version can surpass the $1.5$ barrier in the case of $m=3$; the proof is provided in \Cref{sec:m=3}. The original \SLEEPY algorithm~\cite{DBLP:journals/tcs/NogaS01} is in fact a special instance of our framework for $m=2$, obtained by setting $\lambda = 1$ (thus, no dynamic locking) and choosing $\alpha = \frac{3 - \sqrt{5}}{2}$. As a result, \GSLEEPY attains the same optimal competitive ratio of $\frac{5 - \sqrt{5}}{2}$ as \SLEEPY when there are two machines.

To analyze the competitive ratio of the algorithm, we aim to show that no job set $\mathcal{J}$ can cause the algorithm to produce a makespan satisfying $\sigma(\mathcal{J}) > (1+\gamma)\pi(\mathcal{J})$. Without loss of generality, we assume for contradiction that there exists a smallest (in terms of the number of jobs) counterexample $\mathcal{J}$ such that $\pi(\mathcal{J}) = 1$ and $\sigma(\mathcal{J}) > (1+\gamma)\pi(\mathcal{J}) = 1 + \gamma$. We label the jobs in $\mathcal{J}$ by indices $1, \dots, n$ according to their completion times in the schedule $\sigma$, so that $C_1 \le C_2 \le \dots \le C_n$.

To derive a contradiction, we focus on two types of arguments: bin-packing arguments and efficiency arguments. The goal of the bin-packing arguments is to identify a sufficiently large subset of large jobs, such as $m+1$ jobs with processing time $p_j > 0.5$, that cannot be feasibly scheduled within a makespan of $1$, contradicting $\pi(\mathcal{J}) = 1$. In contrast, the efficiency arguments aim to show that the algorithm processes a substantial amount of workload, such that it would be impossible for the optimal offline schedule to complete all that workload within $[0,1)$. As a simple example, if the total workload processed by the algorithm exceeds $m$, then a contradiction arises. However, since not all jobs are released at time $0$, the upper bound of processed workload by the optimal solution may be less than $m$. We must employ more refined arguments that account for release times in our analysis.

\subsection{Algorithm's Efficiency}

In this subsection, we introduce some useful notations to describe the efficiency of an algorithm, as preparation for applying efficiency arguments. Roughly speaking, efficiency measures the total workload completed by the algorithm within a given time interval. We also present a technical lemma from \cite{DBLP:conf/approx/0002KTWZ18} that provides an upper bound on the efficiency gap between our algorithm and the optimal solution.


\begin{definition}
Consider a schedule $\sigma$. To describe whether a machine is wasting time, we define two indicators for each machine $M \in \mathcal{M}$ as follows:
\begin{itemize}
\item $\mathbbm{1}^P_\sigma(M, t)$, which indicates that machine $M$ is busy in schedule $\sigma$ at time $t$.
\item $\mathbbm{1}^W_\sigma(M, t)$, which indicates that in schedule $\sigma$, at time $t$, machine $M$ is not busy even though there are pending jobs.
\end{itemize}
It is worth noting that, under the strategy used by \GSLEEPY, the event corresponding to $\mathbbm{1}^W_\sigma(M, t)$ occurs only when a machine is locked.
\end{definition}

\begin{definition}\label{def:WPPhat}
    For an interval $[t_1,t_2)$, we define the following notations to measure the efficiency of $\sigma$ in this period. 
    \begin{itemize}
        \item $P_\sigma(t_1,t_2)$: the total processing time (workload) in $\sigma$ during $(t_1,t_2)$, i.e., 
        \[
        P_\sigma(t_1,t_2) \triangleq \sum_{M\in\mathcal{M}}\int_{t_1}^{t_2}\mathbbm{1}^P_\sigma(M, t)dt.
        \]
        \item $W_\sigma(t_1,t_2)$: the total waste of processing power in $\sigma$ during $(t_1,t_2)$, i.e., \[
        W_\sigma(t_1,t_2)  \triangleq \sum_{M\in\mathcal{M}}\int_{t_1}^{t_2} \mathbbm{1}^W_\sigma(M, t)dt.
        \]
        \item $\hat{P}_\sigma(t_1,t_2)$: the total extended  processing time in $\sigma$ during $(t_1,t_2)$ (including the unfinished parts of jobs that are being processed at time $t_2$), i.e., \[
        \hat{P}_\sigma(t_1,t_2) \triangleq P_{\sigma}(t_1,t_2) + \sum_{j: t_1 \leq s_j < t_2, C_j>t_2} (C_j - t_2) = \sum_{j:t_1\le s_j<t_2}p_j+\sum_{j:s_j<t_1<C_j}(C_j-t_1).
        \]
    \end{itemize}
\end{definition}
For convinence, we use $P_\pi$ and $P_\sigma$ denote $P_\pi(0,1)$ and $P_\sigma(0,C_n)$, which means the total processing time of $\sigma$ and $\pi$. Naturally, we have $P_\pi = P_\sigma \leq m$. 

The following lemma provides the basic efficiency argument for \GSLEEPY, with or without dynamic locking. It gives an upper bound on the total processing power wasted by the algorithm.

\begin{lemma}
\label{lem:boundwaste}
The total waste of processing power during $(t_1,t_2)$ is upper bounded using the extended processing time during $(t_1,t_2)$, i.e. $W_\sigma(t_1,t_2)\le(m-1)\alpha\hat{P}_\sigma(t_1,t_2)$.
\end{lemma}

\begin{proof}
    According to the algorithm’s definition, each processing job $j$ can create at most $m-1$ locked periods, each extending from $s_j$ to $s_j + \alpha_j p_j$. Now, consider all the locked periods that fall within the interval $(t_1, t_2)$. These locked periods can be associated with processing jobs that complete after time $t_1$. There are two cases to analyze.

    First, if $s_j \geq t_1$, the length of the locked period is at most $\alpha_j p_j$. Second, if $s_j < t_1$, the length of the locked period within $(t_1, t_2)$ is at most $\alpha_j p_j - (t_1 - s_j) < \alpha_j (C_j - t_1)$.
    
    Summing these bounds over all such jobs, we find that the total locked time during $(t_1, t_2)$ is at most $(m-1) \alpha \hat{P}_\sigma(t_1, t_2)$.
\end{proof}

Finally, we compare the efficiency of the algorithm and the optimal solution over a specific time interval $[0, t)$. That is, we aim to compare $P_\pi(0, t)$ and $P_\sigma(0, t)$. $P_\pi(0, t)$ can exceed $P_\sigma(0, t)$ for two reasons. First, $\sigma$ may waste time due to machines being locked. Second, some machines in $\sigma$ may be idle simply because there are no pending jobs.

Intuitively, the first case can be bounded by $W_\sigma(0, t)$, while the second case can occur only over limited portions of the interval $(0, t)$; otherwise, the efficiency of $\pi$ would also be low. The following left-over lemma formalizes these two sources of inefficiency and provides an integrated upper bound on the difference $P_\pi(0, t) - P_\sigma(0, t)$.

\begin{lemma}[Left-over Lemma \cite{DBLP:conf/approx/0002KTWZ18}] 
\label{lem:leftover} 
$P_\pi(0,t)-P_\sigma(0,t) \le \frac{1}{4}mt + W_\sigma(0,t)$, where $m$ is the number of machines.
\end{lemma}

We remark that this lemma works for general scheduling algorithms. In their paper, the restart operation creates wasted time, which is produced by the locking strategy in our algorithm.

%% file: overview.tex
\section{Overview of Our Analysis}
\label{sec:overview}

Our analysis begins by examining why the last job $n$ has such a large completion time, i.e., $C_n > 1 + \gamma$. The only possible reason is that it experiences a significant delay between its release time $r_n$ and its start time $s_n$.

\begin{lemma}\label{lem:last_delay}
$s_n-r_n>\gamma$.
\end{lemma}
\begin{proof}
    Since $s_n+p_n>1+\gamma$ and $r_n+p_n\le 1$, $s_n-r_n> \gamma$.
\end{proof}

\begin{lemma}
\label{lem:delay_no_idle}
$\forall j \in \mathcal{J}$, all machines are not idle during $(r_j,s_j)$ in $\sigma$. 
\end{lemma}
\begin{proof}
    There is always at least one pending job, the job $j$, during $(r_j,s_j)$ in $\sigma$. And the algorithm will immediately process the longest pending job if a machine is idle (not locked and busy). Thus, $\forall j \in \mathcal{J}$, there is no idle machines during $(r_j,s_j)$ in $\sigma$.
\end{proof}

Next, we rule out the case where the last job $n$ is too small. Roughly speaking, if $p_n$ is small, $s_n$ should be large, then the algorithm must have processed a substantial amount of workload over the long interval $[0, s_n)$, which leads to a contradiction via the efficiency arguments. This idea is formalized in the following lemma for the size of the last job. 
\begin{lemma} [Basic Lower Bound on $p_n$]
\label{claim:lowerbound_pn}
We have the following lower bound on $p_n$:
\begin{itemize}
    \item (Used in the general case) $\displaystyle p_n>\frac{m\gamma-\frac{m}{4}-m(m-1)\alpha}{\frac{3}{4}m-1-(m-1)\alpha}$.
    \item (Used in the case of $m=3$) $\displaystyle p_n >\frac{m\gamma-\frac{m}{4}-m(m-1)\alpha+\hat{P}_\sigma(0,s_n)- P_\sigma(0,s_n)}{\frac{3}{4}m-1-(m-1)\alpha}$.
\end{itemize}
\end{lemma}

\begin{proof}
The reason why we can prove the lower bound on $p_n$ is that we observe that $P_{\sigma}$ will include too much workload (larger than $P_{\pi}$) when $p_n$ is small. We divided the time period $[0,C_n]$ into two parts, $[0,r_n]$ and $(r_n,C_n]$ and analyze the difference of workload between $\sigma$ and $\pi$.  

For $[0,r_n]$, after applying \Cref{lem:leftover}, we have
\begin{equation}
\label{eqn:pnlarge_leftover}
    P_\pi(0,r_n)-P_\sigma(0,r_n) \le \frac{m}{4}r_n + W_\sigma(0,r_n).
\end{equation}

We first discuss the algorithm for $(r_n, C_n]$. Recall the definition of $\hat{P}_\sigma(0,s_n)$. It does not include $p_n$ because $n$ starts at $s_n$. Therefore, we have a lower bound on $P_{\sigma}(r_n, C_n)$.
\[
P_{\sigma}(r_n,C_n) \geq P_\sigma(r_n,s_n) + p_n + \hat{P}_\sigma(0,s_n) - P_\sigma(0,s_n).
\]
Remark that $\hat{P}_\sigma(0,s_n) - P_\sigma(0,s_n)$ is a lower bound on workload completed by $\sigma$ after $s_n$ except $p_n$. In general, we will consider the worst case that $\hat{P}_\sigma(0,s_n) - P_\sigma(0,s_n) = 0$. However, when $m=3$, we have some observations to give a better lower bound for it. So, we keep this term in the inequality. 
By \Cref{lem:delay_no_idle}, machines cannot be idle during $(r_n,s_n)$ in $\sigma$. Therefore, we have 
\[
P_{\sigma}(r_n,C_n) \geq m(s_n - r_n) - W_\sigma(r_n,s_n) + p_n + \hat{P}_\sigma(0,s_n) - P_\sigma(0,s_n).
\]
On the other hand, we have a trivial upper bound of $\pi$, which is
\[
P_{\pi}(r_n,C_n) \leq m(1-r_n).
\]
The difference in workload during this period is
\begin{align}
P_{\sigma}(r_n,C_n) - P_{\pi}(r_n,C_n) & \geq  m(s_n - r_n) - W_\sigma(r_n,s_n) + p_n + \hat{P}_\sigma(0,s_n) - P_\sigma(0,s_n) - m(1-r_n)    \notag \\
& =  m(s_n- 1) - W_\sigma(r_n,s_n) + p_n + \hat{P}_\sigma(0,s_n) - P_\sigma(0,s_n) \notag \\ 
& \geq m(s_n + p_n - 1 - p_n)  - W_\sigma(r_n,s_n) + p_n + \hat{P}_\sigma(0,s_n) - P_\sigma(0,s_n) \notag \\
& > m(\gamma - p_n)  - W_\sigma(r_n,s_n) + p_n + \hat{P}_\sigma(0,s_n) - P_\sigma(0,s_n).\label{eqn:pnlarge_secondperiod}
\end{align}
Combining two cases (\Cref{eqn:pnlarge_leftover} and \Cref{eqn:pnlarge_secondperiod}) together, we have a lower bound on $P_{\sigma} - P_{\pi}$.
\begin{align*}
P_{\sigma} - P_{\pi} &> - \frac{m}{4}r_n - W_\sigma(0,r_n) + m(\gamma - p_n)  - W_\sigma(r_n,s_n) + p_n + \hat{P}_\sigma(0,s_n) - P_\sigma(0,s_n) \\ 
&= - \frac{m}{4}r_n + m(\gamma - p_n)  - W_\sigma(0,s_n) + p_n + \hat{P}_\sigma(0,s_n) - P_\sigma(0,s_n). 
\end{align*}
Next, we aim to bound $W_\sigma(0,s_n)$. By the basic efficiency of the algorithm (\Cref{lem:boundwaste}), we have 
\[
W_\sigma(0,s_n) \leq (m-1)\alpha \hat{P}_\sigma(0,s_n). 
\]
Further, because $\hat{P}_\sigma(0,s_n)$ does not include $p_n$, $\hat{P}_\sigma(0,s_n) \leq m - p_n$, we have 
\[
W_\sigma(0,s_n) \leq (m-1)\alpha (m-p_n).
\]
Therefore, by using this upper bound of $W_\sigma(0,s_n)$, together with $r_n \leq 1 - p_n$ (because otherwise $p_n$ cannot be completed before $1$), we have
\begin{align*}
P_{\sigma} - P_{\pi} &> - \frac{m}{4}r_n + m(\gamma - p_n)  - W_\sigma(0,s_n) + p_n + \hat{P}_\sigma(0,s_n) - P_\sigma(0,s_n) \\
&\geq - \frac{m}{4}(1-p_n) + m(\gamma - p_n)  - (m-1)\alpha (m-p_n) + p_n + \hat{P}_\sigma(0,s_n) - P_\sigma(0,s_n).
\end{align*}
Let 
\begin{align*}
f(p_n) & = - \frac{m}{4}(1-p_n) + m(\gamma - p_n)  - (m-1)\alpha (m-p_n) + p_n + \hat{P}_\sigma(0,s_n) - P_\sigma(0,s_n) \\ 
& = (-\frac{3}{4}m + 1 + (m-1)\alpha)p_n - \frac{m}{4} + m\gamma  - m(m-1)\alpha + \hat{P}_\sigma(0,s_n) - P_\sigma(0,s_n).
\end{align*}
Because $m\geq 2$ and $\alpha <0.5$, the coefficient of $p_n$ is negative. Hence, $f(p_n)$ is increasing when $p_n$ decrease. Therefore, $p_n$ cannot be too small, since we have $f(p_n) \leq P_{\sigma} - P_{\pi} = 0$. In particular, 
\[
p_n > \frac{m\gamma-\frac{m}{4}-m(m-1)\alpha+\hat{P}_\sigma(0,s_n) - P_\sigma(0,s_n)}{\frac{3}{4}m-1-(m-1)\alpha}.
\]
For general case, we use $\hat{P}_\sigma(0,s_n) - P_\sigma(0,s_n) \geq 0$. The inequality becomes
\[
p_n > \frac{m\gamma-\frac{m}{4}-m(m-1)\alpha}{\frac{3}{4}m-1-(m-1)\alpha}.
\]
\end{proof}

To handle the remaining cases when $p_n$ is large, our analysis establishes a series of inequalities involving the parameters $\lambda$, $\alpha$ (the locking parameters), and $\gamma$. These properties play a key role in reaching a contradiction and impose specific constraints on the choice of $\lambda$, $\alpha$, and $\gamma$. Ultimately, by setting $\alpha = \frac{1}{4m^2}$, $\lambda = 4^{25/6}$, and $\gamma = \frac{1}{2} - \frac{1}{4^{20}m^2}$, we arrive at the desired contradiction.
While some of these conditions may appear complex or unintuitive, and only a few turn out to be tight when the parameters are optimized, we have explicitly stated all such conditions within the corresponding properties in \Cref{sec:bound}. This not only allows readers to verify the correctness of our analysis but also sets the stage for a more transparent discussion on parameter selection. 

The next step in our analysis is to examine what happens in the algorithm after time $r_n$, to understand why the last job $n$ is delayed. We refer to this interval as the \emph{last scheduling period}. In \Cref{sec:basic}, we first build a structural understanding of this period within the algorithm’s schedule.

Roughly speaking, this period can be divided into two parts: a \emph{last locking chain} and a \emph{busy period}. The last locking chain consists of a sequence of jobs that lock the next one in a cascading manner, from $n_1 = n$ to $n_k$, until we reach a job $n_k$ that is no longer delayed due to locking. We then show that $s_{n_k}$ is still greater than $r_n$, meaning there is a busy period $[r_n, s_{n_k})$ that delays job $n_k$. Since all machines are busy during this interval, there must be $m$ jobs processing on these machines. We refer to these $m$ jobs, together with the $k$ jobs in the last locking chain, as the \emph{critical jobs}. Refer to \Cref{fig:locking_chain&critical_jobs_pre} as an example.

\begin{figure}[htbp]
    \centering
    \includegraphics[width=0.5\textwidth]{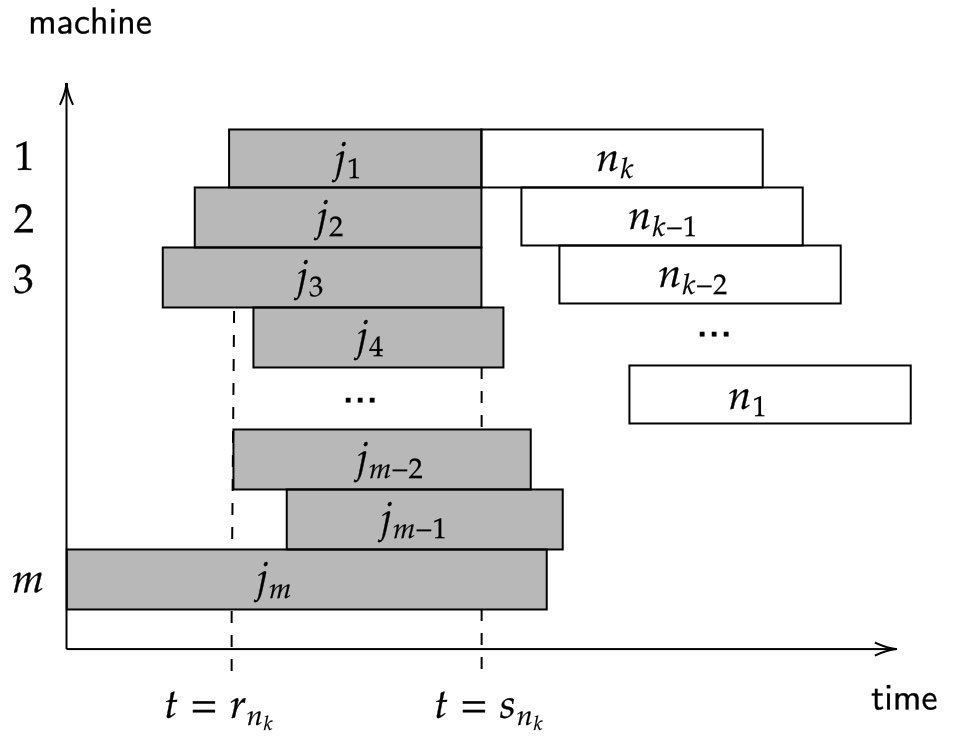}
    \caption{Locking chain (the white jobs) and critical jobs (the white and gray jobs).}
    \label{fig:locking_chain&critical_jobs_pre}
\end{figure}

We categorize these critical jobs into two types: 
\begin{itemize}
    \item early critical jobs, where $s_j < r_n$;
    \item late critical jobs, where $s_j \geq r_n$.
\end{itemize}

By proving that no three critical jobs can be scheduled on the same machine, we conclude, via the Pigeonhole Principle, that there must be $k$ pairs of critical jobs scheduled on the same machine in $\pi$. We then distinguish two cases:
\begin{itemize}
    \item If all these $k$ pairs consist solely of early critical jobs, we have $k$ machines, each processing an early critical pair. We handle this case using a generalized bin-packing argument and resolve it by carefully setting the rescaling factor $\lambda$ in \Cref{sec:earlypair}. Roughly speaking, we show that the total processing time of these critical jobs is at least $k$, so at least one machine must have a makespan exceeding $1$, contradicting $\pi = 1$.
    \item Otherwise, if at least one pair involves a late critical job, we apply the efficiency argument by comparing the efficiency of our algorithm with that of the optimal solution. This analysis is detailed in \Cref{sec:case1} and \Cref{sec:case2}, where we derive a contradiction.
\end{itemize}

Finally, we conclude the proof of \Cref{thm:ratio_m>3} with a sketch. Moreover, we also note why we can achieve a better ratio when $m=3$ even without dynamic locking. Using an efficiency argument, we can show (in \Cref{claim:nispushed}) that the last locking chain can consist only of the last job $n$, which significantly simplifies the challenge posed by the last locking chain.

\ratiogeneral*

\begin{proof}
The proof combines the following key lemmas:

\begin{itemize}
    \item \Cref{claim:no3} in \Cref{sec:basic}: we prove that no three critical jobs can be scheduled on the same machine in $\pi$.
    \item \Cref{claim:sasbrnk} in \Cref{sec:earlypair}: we show that it is impossible for all critical pairs to consist of early critical jobs.
    \item \Cref{claim:sasbearly} in \Cref{sec:case1} and \Cref{lem:latecase2main} in \Cref{sec:case2}: these two subcases handle situations where a critical pair includes a late critical job.
\end{itemize}

We emphasize again that these lemmas are conditional results, holding only when certain listed conditions are satisfied. In \Cref{sec:bound}, we show that by setting $\alpha = \frac{1}{4m^2}$, $\lambda = 4^{25/6}$, and $\gamma = \frac{1}{2} - \frac{1}{4^{20}m^2}$, all these conditions indeed hold. This establishes that our algorithm is $(1.5 - \frac{1}{O(m^2)})$-competitive.
\end{proof}

%% file: analysis_criticaljob.tex
\section{Basic Discussion: Last Scheduling Period}
\label{sec:basic}
This section discusses the critical period in $\sigma$, i.e., the period after $r_n$. We aim to understand why the last job $n$ is delayed until time $s_n$. The delay of the last job in our setting can occur for two different reasons.

On the one hand, the job may be delayed because all machines are busy. In this case, we call it a pushed job. On the other hand, the delay may occur because machines are locked, and the job is scheduled immediately once a machine becomes unlocked. We refer to such jobs as locked jobs. Note that in a corner-case scenario where a job is released at the same moment a machine becomes unlocked and is immediately scheduled, we still classify it as a locked job.

Now, we scan backward starting from the last job, with the goal of finding the first job in this sequence that is not locked. If the last job $n_1 = n$ is locked, we identify the job responsible for locking the machines, denoted as $n_2$. If $n_2$ is also a locked job, we repeat this process until we find a job $n_k$ that is not a locked job (i.e., it is either pushed or scheduled immediately). We define the set $S = \{n_1 = n, n_2, \dots, n_k\}$ as the \textbf{last locking chain} in $\sigma$.

By the definition of the locking chain, we have the following claim.

\begin{proposition}
\label{claim:lockingchain_basic}
For all $2\leq i \leq k$, $s_{n_{i-1}} = s_{n_{i}} + \alpha_{n_{i}} p_{n_{i}}$. 
\end{proposition}
Then, we have the following advanced property under a specific condition of the three parameters $\lambda$, $\alpha$, and $\gamma$. 

\begin{definition}\label{def:af}
    We denote by $\alpha_f$ the largest locking parameter of jobs in $S$, i.e. $\alpha_f\triangleq\max_{j\in S}\alpha_j$.
\end{definition}
\begin{proposition}\label{lem:finallockingchain}
   We have the following advanced properties for the final locking chain:
    \begin{enumerate}
        \item When $2 \leq k \leq m + 1$, for all $1 \leq i \leq k - 1$, we have:
        \[
        p_{n_{i+1}} \geq p_{n_i} \quad \text{and} \quad s_{n_{i+1}} - r_{n_{i+1}} > \gamma - \alpha_f \sum_{j=2}^{i+1} p_{n_j}.
        \]
        \item The chain length satisfies $|S| = k \leq m$.
    \end{enumerate}
\end{proposition}

\begin{proof}
We prove the first property for $2 \leq k \leq m+1$ by induction on $k$, and then prove the second property by contradiction.

For the base case $k=2$, by \Cref{lem:last_delay} and the definition of the last locking chain, we have
\[
s_{n_1} - s_{n_2} = \alpha_{n_2} p_{n_2} \leq \alpha \leq \gamma < s_{n_1} - r_{n_1}.
\]
This implies $s_{n_2} > r_{n_1}$, which yields $p_{n_2} \geq p_{n_1}$ by the LPT strategy of our algorithm. Next, we show the delay bound:
\begin{align*}
s_{n_2} - r_{n_2} 
&= s_{n_1} - \alpha_{n_2} p_{n_2} - r_{n_2} \\
&\geq (s_{n_1} + p_{n_1}) - \alpha_{n_2} p_{n_2} - (r_{n_2} + p_{n_2}) 
\end{align*}
Since $s_{n_1} + p_{n_1} > 1 + \gamma$ and $r_{n_2} + p_{n_2} \leq 1$, it follows that
\[
s_{n_2} - r_{n_2} > \gamma - \alpha_{n_2} p_{n_2} \geq \gamma - \alpha_f p_{n_2}.
\]

When $3 \leq k \leq m+1$, we first note:
\[
s_{n_1} - s_{n_{k-1}} = \sum_{j=2}^{k-1} \alpha_{n_j} p_{n_j} \leq (k - 2)\alpha_f.
\]
Assume by the induction hypothesis that the statement holds for all $i < k - 1$. Then,
\begin{align*}
    s_{n_{k-1}} - r_{n_{k-1}} 
    &> \gamma - \alpha_f \sum_{j=2}^{k-1} p_{n_j}  \\
    &\geq \gamma - (k - 2)\alpha_f  \\
    &> \gamma - (m - 1)\alpha_f.
\end{align*}
Under the condition
\begin{empheq}[box=\shadowbox*]{equation}
    \label{condition:lockingchain}
    \frac{\gamma}{\alpha} > m~, \tag{Condition 1}
\end{empheq}
which implies $\frac{\gamma}{\alpha_f} \geq \frac{\gamma}{\alpha} > m$, we have
\[
r_{n_{k-1}} < s_{n_{k-1}} - \gamma + (m-1)\alpha_f < s_{n_{k-1}} - \alpha_f \leq s_{n_k},
\]
where the last inequality follows from the fact that $s_{n_{k-1}} = s_{n_k} + \alpha_{n_k} p_{n_k} \leq s_{n_k} + \alpha_f$. Thus, by the algorithm definition, $p_{n_k} \geq p_{n_{k-1}}$.

We now analyze the delay of $n_k$:
\begin{align*}
    s_{n_k} - r_{n_k} 
    &= (s_{n_{k-1}} - \alpha_{n_k} p_{n_k}) - r_{n_k} \\
    &\geq (s_{n_{k-1}} - \alpha_{n_k} p_{n_k} + p_{n_{k-1}}) - (r_{n_k} + p_{n_k}) 
    \quad \tag{since $p_{n_k} \geq p_{n_{k-1}}$} \\
    &\geq (s_{n_{k-1}} - \alpha_{n_k} p_{n_k} + p_{n_1}) - (r_{n_k} + p_{n_k}) 
    \quad \tag{since $p_{n_{k-1}} \geq p_{n_1}$ by induction hypothosis} \\
    &= (s_{n_1} - \sum_{j=2}^{k-1} \alpha_{n_j} p_{n_j} - \alpha_{n_k} p_{n_k} + p_{n_1}) - (r_{n_k} + p_{n_k}) 
    \quad \tag{by unrolling $s_{n_{k-1}}$} \\
    &\geq (s_{n_1} - \sum_{j=2}^{k} \alpha_{n_j} p_{n_j} + p_{n_1}) - 1 
    \quad \tag{as $r_{n_k} + p_{n_k} \leq 1$ since $\pi = 1$} \\
    &> \gamma - \sum_{j=2}^{k} \alpha_{n_j} p_{n_j} 
    \quad \tag{as $\sigma = C_n = s_{n_1} + p_{n_1} > 1 + \gamma$} \\
    &\geq \gamma - \alpha_f \sum_{j=2}^{k} p_{n_j}.
\end{align*}

Consequently, the first property holds for all $2 \leq k \leq m+1$ by this inductive proof. 

Next, we prove the second property by contradiction. Suppose $k \geq m + 1$. Then, among the $k$ jobs in the chain, at least two must be processed on the same machine. That is, there exist indices $i < j$ such that both $n_i$ and $n_j$ are assigned to the same machine, and hence
\[
C_{n_j} = s_{n_j} + p_{n_j} \leq s_{n_i}.
\]

Among all such pairs $(i, j)$, choose the one with the smallest $j$. By the pigeonhole principle, we have $j \leq m + 1$. Since the jobs form a locking chain, it holds that
\[
s_{n_i} - s_{n_j} = \sum_{k = i+1}^{j} \alpha_{n_k} p_{n_k}, \quad \text{and} \quad p_{n_j} \geq p_{n_{j-1}} \geq \dots \geq p_{n_{i+1}}.
\]
Then we estimate:
\begin{align*}
s_{n_j} + p_{n_j} 
&= s_{n_i} - \sum_{k = i+1}^{j} \alpha_{n_k} p_{n_k} + p_{n_j} \\
&\geq s_{n_i} - (j - i)\alpha_f p_{n_j} + p_{n_j} \\
&= s_{n_i} + \left(1 - (j - i)\alpha_f\right) p_{n_j} \\
&\geq s_{n_i} + (1 - m\alpha_f) p_{n_j} > s_{n_i},
\end{align*}
where the last inequality follows from \eqref{condition:lockingchain}, which implies $m\alpha_f < 1$. 
This contradicts the assumption that $s_{n_j} + p_{n_j} \leq s_{n_i}$. Therefore, no two jobs from $S$ are processed on the same machine. Since there are only $m$ machines, we conclude that $|S| \leq m$.
\end{proof}
\begin{definition}\label{def:gamma'}
    For simplicity, define $\gamma'\triangleq \gamma-\alpha_f\sum_{j=2}^{k} p_{n_j}$.
\end{definition}
Then, we have the following corollaries.
\begin{corollary}\label{cor:Cnk}
    $s_{n_k}+p_{n}>1+\gamma'$.
\end{corollary}
\begin{proof}
For $i = 1, \dots, k$, we have 
$s_{n_k} + p_n = s_{n_1} - \displaystyle\sum_{j=2}^k \alpha_{n_j} p_{n_j} + p_n > 1 + \gamma - \alpha_f \displaystyle\sum_{j=2}^k p_{n_j} = 1 + \gamma'$ 
by \Cref{def:gamma'} and $s_n + p_n > 1 + \gamma$.
\end{proof}
\begin{corollary}\label{cor:snk}
    $s_{n_k}>\gamma'+r_n$.
\end{corollary}
\begin{proof}
    By \Cref{claim:lockingchain_basic}, we have $s_{n_k} = s_n - \displaystyle\sum_{j=2}^k \alpha_{n_j} p_{n_j}$, and by \Cref{def:af}, $s_{n_k} \ge s_n - \alpha_f \displaystyle\sum_{j=2}^k p_{n_j}$. Then, using \Cref{def:gamma'} and \Cref{lem:last_delay}, we obtain $s_{n_k} \ge s_n + \gamma' - \gamma > \gamma' + r_n$.
\end{proof}

Next, since $s_{n_k} > r_n$ and $n_k$ is not a locked job by definition, it follows that there must be $m$ jobs being processed at time $s_{n_k}$. Otherwise, it would be impossible to delay job $n$ to a start time $s_n > s_{n_k}$. We denote these jobs by $j_1, j_2, \dots, j_m$ and refer to them as \emph{\textbf{critical jobs}}, together with the jobs in the locking chain. That is, we define the critical job set as
\[
\mathcal{J}_C = \{j_1, \dots, j_m\} \cup \{n_1, \dots, n_k\}.
\]
Refer to \Cref{fig:locking_chain&critical_jobs} for illustration.

\begin{figure}[H]
    \centering
    \includegraphics[width=0.5\textwidth]{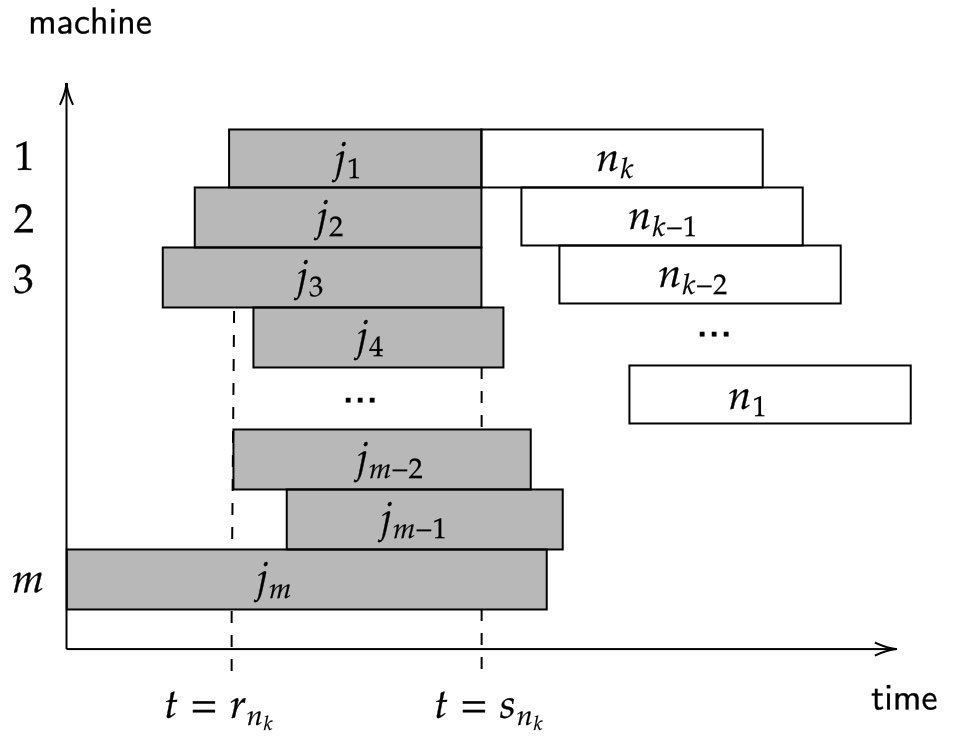}
    \caption{Locking chain (the white jobs) and critical jobs (the white and gray jobs).}
    \label{fig:locking_chain&critical_jobs}
\end{figure}

\begin{definition}[Early and Late Critical Jobs]
\label{lem:earlyandlate}
For a critical job $j$, we call it
\begin{itemize}
    \item \textbf{Early Critical Job}, if $s_{j} < r_{n}$;
    \item \textbf{Late Critical Job}, if $s_{j} \geq r_{n}$. 
\end{itemize}
\end{definition}
\begin{lemma}
    \label{lem:critical_large}
    If a critical job $j$ is an early critical job, then $p_{j} > \gamma'$. Otherwise, if it is a late critical job, $p_{j} \geq p_n$.
\end{lemma}
\begin{proof}
If $s_j<r_n$, $p_j\ge s_{n_k}-s_j>s_{n_k}-r_n>\gamma'$. If $s_j\ge r_n$, then $p_j\geq p_n$ because of the LPT strategy.    
\end{proof}

\begin{lemma}
\label{claim:no3}
No $3$ jobs $\in \mathcal{J}_C$ are executed in the same machine in $\pi$.
\end{lemma}

\begin{proof}
    With condition
    \begin{empheq}[box=\shadowbox*]{equation}
    \label{condition:gamma'}
        \gamma-(m-1)\alpha\ge0.4,     
    \tag{Condition 2}  
    \end{empheq}
    and the definition of $\gamma'$, we have $\gamma'=\gamma-\alpha_f\sum_{j=2}^kp_{n_j}\geq\gamma-(m-1)\alpha\geq0.4$.
    With condition
    \begin{empheq}[box=\shadowbox*]{equation}
    \label{condition:pnlarge}
        \frac{m\gamma-\frac{m}{4}-m(m-1)\alpha}{\frac{3}{4}m-1-(m-1)\alpha}>\frac{1}{3}, 
        \tag{Condition 3}
    \end{empheq}
    and \Cref{claim:lowerbound_pn}, we have $p_n>\frac{1}{3}$. 
    Since all jobs in $\mathcal{J}_C$ have processing time at least either $p_n$ or $\gamma'$ (by \Cref{lem:critical_large}), and both are at least $\frac{1}{3}$ due to the condition, it is impossible to arrange $3$ jobs $\in \mathcal{J}_C$ on the same machine in $\pi$.
\end{proof}

Now, we have completed the preparation for our analysis. In the critical period ($t > r_n$), we identify $m + k$ critical jobs. By \Cref{claim:no3} and the pigeonhole principle, at least $k$ pairs of these jobs must be executed on the same machine in $\pi$. We refer to each such pair as a \textbf{critical job pair}.

%% file: analysis_eff_pre.tex
\section{Further Discussion of Critical Job Pairs}
\label{sec:critical}
In this section, we analyze different cases among the $k$ critical job pairs. We classify them into two types:
\begin{itemize}
    \item \textbf{Early Critical Pair}: both jobs are early critical jobs.
    \item \textbf{Late Critical Pair}: at least one job is a late critical job.
\end{itemize}

We first eliminate the case where all critical pairs are early in \Cref{sec:earlypair}, which is a key scenario for the dynamic locking technique. In this section, we will also fix the choice of $\lambda$. Then, in \Cref{sec:case1} and \Cref{sec:case2}, we analyze the efficiency of both our algorithm and the optimal solution in the presence of at least one late critical pair, ultimately leading to a contradiction.

\subsection{Early Critical Pairs Only}
\label{sec:earlypair}
This section eliminates the case that all critical job pairs consist of early critical jobs. First, we give a lower bound for $\alpha_f$, which promises that the critical jobs will use a large enough locking parameter. 
\begin{proposition}
\label{claim:aflow}
$\alpha_f>\lambda^{-4.5}\alpha$.
\end{proposition}
\begin{proof}
    First, we prove $s_n \leq 1 + \gamma$. If $s_n = r_n$, we have $s_n = r_n < 1$. Otherwise, either it is locked, i.e., $\exists j, s_n = s_j + \alpha_j p_j \leq C_j \leq C_{n-1}$, or it is pushed, i.e., $\exists j, s_n = C_j \leq C_{n-1}$. Finally, we conclude $s_n \leq 1+ \gamma$ by showing $C_{n-1} \leq 1+\gamma$. The fact comes from that $\mathcal{J}$ is minimum.
    If there is a job set $\mathcal{J}=\{n-k,\dots,n\}(k\ge 1)$ with $\forall j \in \mathcal{J}, C_j>1+\gamma$, then delete the job with the latest start time in $\mathcal{J}$ from $\mathcal{J}$. Denote the new job set as $\mathcal{J}'$. We have $\sigma(\mathcal{J}')=\sigma(\mathcal{J}) > 1+\gamma$ and $\pi(\mathcal{J}')\le \pi(\mathcal{J})=1$. It contradicts the notion that $\mathcal{J}$ is the minimum. Therefore, we have $s_n \leq 1 + \gamma$.
    
    Then, by \Cref{claim:lowerbound_pn} and \eqref{condition:pnlarge}, we have $p_n > \frac{1}{3}$. It implies that $\frac{s_n}{p_n}< 3+3\gamma\le4.5$. Thus, 
    \[
    \alpha_f=\max_{j\in S}\alpha_j=\max_{j\in S}\lambda^{-\frac{s_j}{p_j}}\alpha\ge\lambda^{-\frac{s_n}{p_n}}\alpha>\lambda^{-4.5}\alpha.
    \]
\end{proof}
Hereafter, we show that with a proper choice of $\lambda$, we can eliminate the possibility of all critical job pairs consisting of early critical jobs. 
\begin{lemma}
\label{claim:sasbrnk}
Fix $\lambda^{0.24} = 4$. There exists at least one late critical pair, i.e., at least one critical job pair ($a$ and $b$) satisfies that either $s_a\ge r_{n}$ or $s_b\ge r_{n}$.
\end{lemma}

\begin{proof}
Assume that every critical job pair $a$ and $b$ satisfies that $s_a< r_{n}$ and $s_b< r_{n}$. In the analysis below, we can conclude that $\pi(\mathcal{J})>1$, which leads to a contradiction. 

By definition, these $2k$ critical jobs are all in $\{j_1,\dots,j_m\}$. A straightforward observation is that $2k\leq m$. Identify them with integer $l_1,\dots,l_{2k}$ such that $s_{l_1}<\dots<s_{l_{2k}}$. Define $\alpha_s\triangleq \min_{j\in \{l_1,\dots,l_{2k}\}}\alpha_j$. The total processing time of these jobs is at least 
\[2k(s_{n_k}-s_{l_{2k}})+(2k-1)\alpha_sp_{l_{2k-1}}+\dots+\alpha_sp_{l_1}>2k\gamma'+k(2k-1)\alpha_s\gamma'.\]
Since these jobs are processed on $k$ machines in $\pi$, the total processing time of these jobs is at most $k$. Otherwise, $\pi(\mathcal{J})>1$. However, we can get a lower bound of $\alpha_s$ as follows.

By the definition of $\gamma'$ and \eqref{condition:gamma'}, we have $\gamma' \ge \gamma-(m-1)\alpha\ge 0.4$. For job $j\in \{l_1,\dots,l_{2k}\}$, we have $s_j < r_{n}$ and $p_j > \gamma'$. Then, 
\begin{equation}
    \frac{s_j}{p_j}<\frac{r_{n}}{\gamma'}.
    \label{sj/pj}
\end{equation}
For job $n_i\in \{n_1,\dots,n_k\}$, by \Cref{cor:snk}, $s_{n_i}\geq s_{n_k} > \gamma'+ r_{n}$, and $p_i \leq p_{n_k} < p_n + \gamma-\gamma'\le1-r_n+\gamma-\gamma'$. Then,
\begin{equation}
    \frac{s_{n_i}}{p_{n_i}}\ge\frac{s_{n_k}}{p_{n_k}}>\frac{\gamma'+ r_{n}}{1-r_n+\gamma-\gamma'}.
    \label{sni/pni}
\end{equation}
By definition of $\alpha_s$ and $\alpha_f$, with \Cref{sj/pj} and \Cref{sni/pni}, we obtain that 
\begin{align*}
    \frac{\alpha_s}{\alpha_f}>\lambda^{-(\frac{r_{n}}{\gamma'}-\frac{\gamma'+ r_{n}}{1-r_n+\gamma-\gamma'})}.
\end{align*}
Let $f(r_n)\triangleq-\frac{r_{n}}{\gamma'}+\frac{\gamma'+ r_{n}}{1-r_n+\gamma-\gamma'}$. Then, 
\[\frac{\partial f}{\partial r_n}=-\frac{1}{\gamma'}+\frac{1+\gamma}{(1-r_n+\gamma-\gamma')^2}.\]
We get $f'(r_n)<0$ when $r_n\in[0,1+\gamma-\gamma'-\sqrt{\gamma'(1+\gamma)})$ and $f'(r_n)>0$ when $r_n\in(1+\gamma-\gamma'-\sqrt{\gamma'(1+\gamma)},1)$. Therefore, 
\[\frac{\alpha_s}{\alpha_f}>\lambda^{f(r_n)}\geq \lambda^{f(1+\gamma-\gamma'-\sqrt{\gamma'(1+\gamma)})}=\lambda^{2\sqrt{\frac{1+\gamma}{\gamma'}}-\frac{1+\gamma}{\gamma'}}.\]
Conditioned on 
\begin{empheq}[box=\shadowbox*]{equation}
    \gamma\ge\frac{1+3.5(m-1)\alpha}{2.5},
    \tag{Condition 4}
\end{empheq}
we can get $\frac{1+\gamma}{\gamma'}\leq \frac{1+\gamma}{\gamma-(m-1)\alpha}\leq3.5$ and 
$\frac{1+\gamma}{\gamma'}\geq\frac{1}{\gamma'}+1>3$. Let $h(x)\triangleq2\sqrt{x}-x$. We know $h(x)$ strictly decreases in $[1,+\infty)$. Then we have that $2\sqrt{\frac{1+\gamma}{\gamma'}}-\frac{1+\gamma}{\gamma'}\geq 2\sqrt{3.5}-3.5>0.24$. Therefore, $\frac{\alpha_s}{\alpha_f}>\lambda^{0.24}=4$. Thus,
\begin{align*}
2k\gamma'+k(2k-1)\alpha_s\gamma'&=k(2+(2k-1)\alpha_s)(\gamma-\alpha_f\sum_{j = 2}^{k}p_{n_j})\\
&\ge k(2+(2k-1)\alpha_s)(\gamma-(k-1)\alpha_f)\\
&>k(2+4(2k-1)\alpha_f)(\gamma-(k-1)\alpha_f).
\end{align*}
By Pigeonhole Principle, there is one machine in $\pi$ have the processing time at least $(2+4(2k-1)\alpha_f)(\gamma-(k-1)\alpha_f)$. Let $g(k) \triangleq (2+4(2k-1)\alpha_f)(\gamma-(k-1)\alpha_f)$, as a lower bound of $\pi$. In definition, $1=\pi\geq g(k)$. However, we have 
\[
\frac{\partial g}{\partial k}=(8\gamma-4(4k-3)\alpha_f-2)\alpha_f\ge 2(4\gamma-2(2m-3)\alpha-1)\alpha_f \geq 0,
\]
with one more condition 
\begin{empheq}[box=\shadowbox*]{equation}
    \label{condition:2gamma2}
        4\gamma-2(2m-3)\alpha-1\ge 0. 
        \tag{Condition 5}
\end{empheq}
Since we fix $\lambda^{0.24} = 4$, by \Cref{claim:aflow}, we obtain $\alpha_f>4^{-17.75}\alpha$. Then we know that 
\[
g(k) \geq g(1)=(2+4\alpha_f)\gamma>(2+4^{-17.75}\alpha)\gamma>1,
\]
if we require the condition
\begin{empheq}[box=\shadowbox*]{equation}
    \label{condition:2gamma1}
        (2+4^{-17.75}\alpha)\gamma>1.
        \tag{Condition 6}
\end{empheq}
This is a contradiction.
\end{proof} 

From \Cref{claim:sasbrnk}, we know that there exists at least one pair of critical jobs $a$ and $b$ such that either $s_a\ge r_{n}$ or $s_b\ge r_{n}$. Consider such a pair of critical jobs $a$ and $b$, called a \textbf{Late Critical Pair}. 

\begin{definition}[Idle Period]
    A period is idle if at least one machine remains idle. $[\ts,\theta^+]$ is defined as the final idle period before $s_n$. If there is no idle period before $s_n$, $\ts=\tf=0$.
\end{definition}
By definition, we have a basic property for this idle period.
\begin{proposition}
    \label{lem:idle_basic}
    $\theta^+ \leq \min_{j\in S} r_j$.
\end{proposition}
\begin{proof}
By definition of the last locking chain, there is no idle segment during $[\min_{j\in S}\{r_j\}, s_n]$. Therefore, $\tf \leq \min_{j\in S}\{r_j\}$.
\end{proof}

We will discuss the following two cases of the fixed late critical pair. 
\begin{enumerate}
    \item $\min\{s^\pi_a,s^\pi_b\} < \theta^+$ in \Cref{sec:case1}.
    \item $\min\{s^\pi_a,s^\pi_b\} \geq \theta^+$ in \Cref{sec:case2}.
\end{enumerate}

%% file: analysis_eff_c1.tex
\subsection{Late Critical Pair: $\min \{s^\pi_a,s^\pi_b\} < \tf$}
\label{sec:case1}
In this section, we derive a contradiction that $\sigma$ processes a workload exceeding $m$ to show that this situation cannot occur. First, we eliminate cases where jobs $a$ and $b$ are late critical jobs. 
\begin{proposition}\label{claim:pre_case1}
    Let $c$ be $a$ if $s^\pi_a<s^\pi_b$ and $b$ otherwise. It follows that $s_c < r_{n}$, leading to $p_c > \gamma'$.  
\end{proposition}
\begin{proof}
    Without loss of generality, assume $s^\pi_a<s^\pi_b$. If $s_a\geq r_{n}$, then $s_a\ge r_{n}\ge \tf>s^\pi_a\geq r_a$, which contradicts that $[\ts,\tf]$ is an idle period.
\end{proof}

Finally, additional restrictions on $\alpha$ and $\gamma$ are imposed to demonstrate a contradiction, reinforcing the infeasibility of this scenario.
\begin{lemma}
\label{claim:sasbearly}
The first case, i.e., $\min \{s^\pi_a,s^\pi_b\} < \tf$, is impossible.
\end{lemma}

\begin{proof}
Without loss of generality, assume $s^\pi_a< s^\pi_b$. By \Cref{claim:pre_case1}, $s_a< r_{n}\le s_b$, $p_a>\gamma'$ and $p_b\ge p_{n}$. Then we have that $1\ge s^\pi_a+p_a+p_b\ge s^\pi_a+p_a+p_{n}$. By \Cref{cor:Cnk}, we also have $s_a+p_a+p_{n}> s_{n_k}+p_{n}>1+\gamma'$. After connecting these two inequalities, we get $1-s^\pi_a\ge p_a+p_{n}>1+\gamma'-s_a$. So $s_a - r_a \geq s_a - s^\pi_a > \gamma'$. By the definition of idle, we know that there won't be idle machines during $[r_a,s_a)$ and $[r_{n},s_{n})$. In \Cref{lem:last_delay}, we prove $s_{n}-r_{n}>\gamma$. So the sum of the total waste of processing power and total processing time during $[r_a,s_{n})$ is greater than the sum of the total waste of processing power and total processing time during $[r_a,s_a)$ and $[r_{n},s_{n})$. Then we have that  $W_\sigma(r_a,s_{n})+P_\sigma(r_a,s_{n})\ge m(\gamma'+\gamma)$. Using \Cref{lem:boundwaste}, 
$
W_\sigma(r_a,s_{n})\le (m-1)\alpha\hat{P}_\sigma(r_a,s_{n}).
$
With these two inequalities above, we have that
\[
\hat{P}_\sigma(r_a,s_{n})\ge\frac{m(\gamma'+\gamma)}{1+(m-1)\alpha}.
\]
By the definition of $\hat{P}_\sigma(t_1,t_2)$, we have that $P_\sigma \geq \hat{P}_\sigma(r_a,s_{n})+\sum_{i=1}^k p_{n_i}$. Hence, we get 
\begin{align*}
  P_\sigma &\geq \hat{P}_\sigma(r_a,s_{n})+\sum_{i=1}^k p_{n_i} \geq \frac{m(\gamma'+\gamma)}{1+(m-1)\alpha}+\sum_{i=1}^k p_{n_i}\\
  &=\frac{m}{1+(m-1)\alpha}(2\gamma- \alpha_f\sum_{i=2}^k p_{n_i}) + \sum_{i=1}^k p_{n_i} \\
  &\ge\frac{2m\gamma}{1+(m-1)\alpha}+(1-\frac{m\alpha}{1+(m-1)\alpha})\sum_{i=2}^k p_{n_i}+p_n.
\end{align*}
Conditioned on 
\begin{empheq}[box=\shadowbox*]{equation}
    \label{condition:beforethetacondition1}
        1-\frac{m\alpha}{1+(m-1)\alpha}\ge0,
    \tag{Condition 7}
\end{empheq}
we have 
$P_\sigma \ge\frac{2m\gamma}{1+(m-1)\alpha}+p_n$.
By \Cref{claim:lowerbound_pn}, $p_n>\frac{m\gamma-\frac{m}{4}-m(m-1)\alpha}{\frac{3}{4}m-1-(m-1)\alpha}$. 
With another condition 
\begin{empheq}[box=\shadowbox*]{equation}
    \label{condition:beforethetacondition2}
    \frac{2\gamma}{1+(m-1)\alpha}+\frac{\gamma-\frac{1}{4}-(m-1)\alpha}{\frac{3}{4}m-1-(m-1)\alpha}\ge 1, 
    \tag{Condition 8}
\end{empheq}
we have $P_\sigma\geq\frac{2m\gamma}{1+(m-1)\alpha}+p_n>m(\frac{2\gamma}{1+(m-1)\alpha}+\frac{\gamma-\frac{1}{4}-(m-1)\alpha}{\frac{3}{4}m-1-(m-1)\alpha})\ge m$, which is a contradiction.
\end{proof} 

%% file: analysis_eff_c2_lipics.tex
\subsection{Late Critical Pair: $\min \{s^\pi_a,s^\pi_b\} \ge \tf$}
\label{sec:case2}
In this section, we derive a contradiction that the overall workload of our algorithm (i.e., $p_\sigma$) is larger than that of OPT (i.e., $P_\pi$). First, we prove two technical lemmas concerning the efficiency of our algorithm to derive a lower bound of $P_{\sigma}- P_{\pi}$. Subsequently, we derive a contradiction by demonstrating that the lower bound is positive. 

\begin{proposition}
\label{cor:leftover1}
$\forall t\in [\tf,s_n]$, $P_\pi(0,\tf) + m(t-\tf) -P_\sigma(0,t)\le\frac{m}{4}\tf + W_\sigma(0,t)$.
\end{proposition}

\begin{proof}
Using \Cref{lem:leftover} over time period $(0,\tf)$, we have that 
\begin{align*}
P_\pi(0,\tf)-P_\sigma(0,\tf) \le \frac{m}{4}\tf + W_\sigma(0,\tf).
\end{align*}
By the definition of the last idle period, there is no idle time during $(\tf,s_n)$. So for $t\in [\tf,s_n]$, $ m(t-\tf)-P_\sigma(\tf,t) \le W_\sigma(\tf,t)$. We can get $P_\pi(0,\tf) + m(t-\tf)-P_\sigma(0,t)\le \frac{m}{4}\tf+W_\sigma(0,t)$ after summing these two inequalities above.    
\end{proof}

\begin{corollary}\label{cor:leftover2}
$\forall t\in [\tf,s_n]$, $P_\pi(0,\tf) + m(t-\tf) -\hat{P}_\sigma(0,t)\le \frac{m}{4}\tf + (m-1) m \alpha t$.
\end{corollary}

\begin{proof}
By \Cref{cor:leftover1},
\[P_\pi(0,\tf) + m(t-\tf)-\hat{P}_\sigma(0,t) \le \frac{m}{4}\tf + W_\sigma(0,t) - (\hat{P}_\sigma(0,t) - {P}_\sigma(0,t)).\]
By \Cref{lem:boundwaste},
\[W_\sigma(0,t)\le(m-1)\alpha\hat{P}_\sigma(0, t).\]
Then we have 
\begin{align*}
    P_\pi(0,\tf) + m(t-\tf)-\hat{P}_\sigma(0,t) &\leq \frac{m}{4}\tf + (m-1)\alpha \hat{P}_\sigma(0,t) - (\hat{P}_\sigma(0,t) - {P}_\sigma(0,t))\\
    &=\frac{m}{4}\tf + (m-1) \alpha {P}_\sigma(0,t) + ((m-1)\alpha - 1) (\hat{P}_\sigma(0,t) - {P}_\sigma(0,t))
\end{align*}
Conditioned on 
\begin{empheq}[box=\shadowbox*]{equation}
    \label{condition:afterthetacondition1}
        1-(m-1)\alpha\ge 0     \tag{Condition 9}
\end{empheq}
we have $P_\pi(0,\tf) + m(t-\tf)-\hat{P}_\sigma(0,t)\leq \frac{m}{4}\tf + (m-1) \alpha {P}_\sigma(0,t)\leq \frac{m}{4}\tf + (m-1) m \alpha t$.
\end{proof}
\begin{claim}
    $P_\sigma-P_\pi\geq p_n+\sum_{i=2}^k p_{n_i}-\frac{m}{4}\tf-m(m-1)\alpha s_{n_k}-m(1-s_{n_k})$.
    \label{delta}
\end{claim}
\begin{proof}
From the definition of $\tf$, we know $\tf \leq s_{n_k} \leq s_n$. By \Cref{cor:leftover2}, $P_\pi(0,\tf) + m(s_{n_k}-\tf)-\hat{P}_\sigma(0,s_{n_k})\le \frac{m}{4}\tf+m(m-1)\alpha s_{n_k}$, from which we can yield that 
\begin{equation}
    \hat{P}_\sigma(0,s_{n_k}) - P_\pi(0,\tf) \geq m(s_{n_k}-\tf) - \frac{m}{4}\tf - m(m-1)\alpha s_{n_k}.
    \label{subcase1-1}
\end{equation}
By definition of the last locking chain, the workload of $\sigma$ during $(s_{n_k}, C_n)$, except the workload from jobs that start before $s_{n_k}$ is
\begin{equation}
    P_\sigma(s_{n_k},C_n) - (\hat{P}_\sigma(0,s_{n_k}) - {P}_\sigma(0,s_{n_k})) \geq \sum_{i=1}^k p_{n_i}=p_n+\sum_{i=2}^k p_{n_i}.
    \label{subcase1-2}
\end{equation}
Therefore, with \Cref{subcase1-1} and \Cref{subcase1-2}, 
\begin{align*}
    P_\sigma-P_\pi&= P_\sigma(s_{n_k},C_n) + {P}_\sigma(0,s_{n_k}) - (P_\pi(0,\tf) + P_\pi(\tf,1))\\
    &\geq P_\sigma(s_{n_k},C_n) + {P}_\sigma(0,s_{n_k}) - (P_\pi(0,\tf) + m(1-\tf))\\
    &= P_\sigma(s_{n_k},C_n) - (\hat{P}_\sigma(0,s_{n_k}) - {P}_\sigma(0,s_{n_k})) + (\hat{P}_\sigma(0,s_{n_k}) - P_\pi(0,\tf)) - m(1-\tf)\\
    &\geq p_n+\sum_{i=2}^k p_{n_i} + m(s_{n_k}-\tf) - \frac{m}{4}\tf - m(m-1)\alpha s_{n_k} - m(1-\tf)\\
    &= p_n+\sum_{i=2}^k p_{n_i}-\frac{m}{4}\tf-m(m-1)\alpha s_{n_k}-m(1-s_{n_k}).
\end{align*}

\end{proof}



Then, we move towards two subcases of $a$ and $b$, based on the relationship between $s_a$, $s_b$, and $r_{n}$. Recall that $s_a\geq r_{n}$ or $s_b\geq r_{n}$ is proved in \Cref{claim:sasbrnk}. Without loss of generality, assuming $s_b \geq s_a$, it remains to discuss the following two cases:
\begin{enumerate}
    \item $s_a < r_{n} \leq s_b$: $a$ is early critical job and $b$ is late. 
    \item $r_{n} \leq s_a \leq s_b$: $a$ and $b$ are both late critical jobs.  
\end{enumerate}

\begin{proposition}[Subcase 1]
\label{claim:sasbrnk_case1}
If $\min \{s^\pi_a,s^\pi_b\} \ge \tf$, it is impossible that $s_a < r_{n}\le s_b$.
\end{proposition}

\begin{proof}
We have that $p_{n}\le p_b\le1-p_a-\tf<1-\gamma'-\tf$, from which we can see $\tf<1-\gamma'-p_{n}<s_n-\gamma-\gamma'$. By \Cref{delta}, $P_\sigma-P_\pi\geq p_n+\sum_{i=2}^k p_{n_i}-\frac{m}{4}\tf-m(m-1)\alpha s_{n_k}-m(1-s_{n_k})$. Since $\tf<1-\gamma'-p_{n}$ and $s_{n_k}+p_{n}>1+\gamma'$ (\Cref{cor:Cnk}),

\begin{align*}
    P_\sigma-P_\pi&\geq p_n+\sum_{i=2}^k p_{n_i}-\frac{m}{4}\tf-m(m-1)\alpha s_{n_k}-m(1-s_{n_k})\\
    &> p_n+\sum_{i=2}^k p_{n_i}-\frac{m}{4}(1-\gamma'-p_{n})-m(m-1)\alpha (1+\gamma'-p_{n})-m(p_{n}-\gamma')\\
    &=(1+m(m-1)\alpha-\frac{3}{4}m)p_n+\sum_{i=2}^k p_{n_i}-\frac{m}{4}-m(m-1)\alpha (1+\gamma')+\frac{5}{4}m\gamma'.
\end{align*}
Conditioned on
\begin{empheq}[box=\shadowbox*]{equation}
    \frac{3}{4}m-1-m(m-1)\alpha+1\geq0,
    \tag{Condition 10}
\end{empheq}
the coefficient of $p_n$ is not greater than $0$. Then with $p_{n}<1-\gamma'$, we obtain
\begin{align*}
    P_\sigma-P_\pi&>1-\gamma'+\sum_{i=2}^k p_{n_i}-2m(m-1)\alpha\gamma'-m(1-2\gamma')\\
    &=(1+\alpha-2m\alpha(1-(m-1)\alpha))\sum_{i=2}^k p_{n_i} - (m(1-2\gamma)+(2m(m-1)\alpha+1)\gamma-1).
\end{align*}
Conditioned on 
\begin{empheq}[box=\shadowbox*]{align}
    1+\alpha-2m\alpha(1-(m-1)\alpha)\ge 0,
    \tag{Condition 11}
\end{empheq}
\begin{empheq}[box=\shadowbox*]{align}
    m(1-2\gamma)+(2m(m-1)\alpha+1)\gamma-1\le 0,
    \tag{Condition 12}
\end{empheq}
we get $P_\sigma-P_\pi> 0$. This is a contradiction.
\end{proof} 
\begin{proposition}[subcase 2]
\label{claim:sasbrnk_case2}
If $\min \{s^\pi_a,s^\pi_b\} \ge \tf$, it is impossible that $r_{n} \leq s_a \leq s_b$.
\end{proposition}

\begin{proof}
If $s_a \ge r_{n}$ and $s_b \ge r_{n}$, we have $p_a\ge p_{n}$ and $p_b\ge p_{n}$. 
Then, because $\min \{s^\pi_a,s^\pi_b\} + p_a + p_b \leq 1$, we have $p_a+p_b \leq 1 - \tf$. It implies that $p_{n}\le \frac{p_a+p_b}{2}\leq\frac{1-\tf}{2}$, so we have 
$\tf\le 1-2p_n$. Combining $\tf\leq 1-2p_{n}$, $s_{n_k}+p_{n}>1+\gamma'$ (\Cref{cor:Cnk}) and \Cref{delta}, we get

\begin{align*}
    P_\sigma-P_\pi&\geq p_n+\sum_{i=2}^k p_{n_i}-\frac{m}{4}\tf-m(m-1)\alpha s_{n_k}-m(1-s_{n_k})\\
    &> p_n+\sum_{i=2}^k p_{n_i}-\frac{m}{4}(1-2p_{n})-m(m-1)\alpha (1+\gamma'-p_n)-m(p_n-\gamma')\\
    &=(1+m(m-1)\alpha-\frac{1}{2}m)p_n+\sum_{i=2}^k p_{n_i}-\frac{m}{4}-m(m-1)\alpha (1+\gamma')+m\gamma'.
\end{align*}
Conditioned on
\begin{empheq}[box=\shadowbox*]{equation}
    \frac{1}{2}m-1-m(m-1)\alpha\geq 0,\tag{Condition 13}
\end{empheq}
the coefficient of $p_n$ is not greater than 0. Then with $p_{n}\le\frac{1}{2}$, we obtain
\begin{align*}
    P_\sigma-P_\pi&> \frac{1}{2}+\sum_{i=2}^k p_{n_i}-m(m-1)\alpha (\frac{1}{2}+\gamma')-m(\frac{1}{2}-\gamma')\\
    &=(1-m\alpha(1-(m-1)\alpha))\sum_{i=2}^k p_{n_i}-(m(\frac{1}{2}-\gamma)+m(m-1)\alpha (\frac{1}{2}+\gamma)-\frac{1}{2}).
\end{align*}
Conditioned on 
\begin{empheq}[box=\shadowbox*]{equation}
    1-m\alpha(1-(m-1)\alpha)\ge0,\tag{Condition 14}
\end{empheq}
\begin{empheq}[box=\shadowbox*]{equation}
    m(\frac{1}{2}-\gamma)+m(m-1)\alpha (\frac{1}{2}+\gamma)-\frac{1}{2}\le 0,\tag{Condition 15}
\end{empheq}
we get $P_\sigma-P_\pi> 0$. This is a contradiction.
\end{proof} 

\begin{restatable}{lemma}{secondcase}
\label{lem:latecase2main}
The second case, i.e., $\min \{s^\pi_a,s^\pi_b\} \geq \tf$, is impossible.
\end{restatable} 

\begin{proof}
    It is a straightforward combination of \Cref{claim:sasbrnk_case1} and \Cref{claim:sasbrnk_case2}. 
\end{proof}

%% file: bound.tex
\section{Bound of $\gamma$ and $\alpha$}
\label{sec:bound}

Given the conditions concerning variables $\lambda, \ \alpha, \ \gamma$ and $m$ as mentioned earlier in the text, we prove that there does not exist a minimum size set of jobs $\mathcal{J}$ with optimal makespan $\pi(\mathcal{J})=1$ that makes our algorithm have a makespan $\sigma(\mathcal{J})=C_n=s_n+p_n>1+\gamma$. In other words, for each job set $\mathcal{J}$, our algorithm guarantees that $\frac{\sigma(\mathcal{J})}{\pi(\mathcal{J})}\le 1+\gamma$. The conditions are systematically reorganized and presented below for clarity.
\begin{itemize}
    \item Conditions on $\alpha$ in terms of $m$:
\end{itemize}
\begin{align}
    1-\frac{m\alpha}{1+(m-1)\alpha}\ge0
    \tag{Condition 7}\\
    1-(m-1)\alpha\ge0
    \tag{Condition 9}\\
    \frac{3}{4}m-1-m(m-1)\alpha\ge0
    \tag{Condition 10}\\
    1+\alpha-2m\alpha(1-(m-1)\alpha)\ge 0
    \tag{Condition 11}\\
    -\frac{1}{2}m+(m-1)\alpha+1\le0
    \tag{Condition 13}\\
    1-m\alpha(1-(m-1)\alpha)\ge 0
    \tag{Condition 14}
\end{align}
\begin{itemize}
    \item Conditions on $\gamma$ in terms of $m$ and $\alpha$:
\end{itemize}
\begin{align}
    \frac{\gamma}{\alpha} > m
    \tag{Condition 1}\\
    \gamma-(m-1)\alpha\ge0.4
    \tag{Condition 2}\\
    \frac{m\gamma-\frac{m}{4}-m(m-1)\alpha}{\frac{3}{4}m-1-(m-1)\alpha}>\frac{1}{3}
    \tag{Condition 3}\\
    \gamma\ge\frac{1+3.5(m-1)\alpha}{2.5}
    \tag{Condition 4}\\
    4\gamma-2(2m-3)\alpha-1\ge 0
    \tag{Condition 5}\\
    (2+4^{-17.75}\alpha)\gamma>1
    \tag{Condition 6}\\
    \frac{2\gamma}{1+(m-1)\alpha}+\frac{\gamma-\frac{1}{4}-(m-1)\alpha}{\frac{3}{4}m-1-(m-1)\alpha}\ge 1
    \tag{Condition 8}\\
    m(1-2\gamma)+(2m(m-1)\alpha+1)\gamma-1\le 0    
    \tag{Condition 12}\\
    m(\frac{1}{2}-\gamma)+m(m-1)\alpha (\frac{1}{2}+\gamma)-\frac{1}{2}\le 0  
    \tag{Condition 15}
\end{align}
Due to Conditions 12 and 15, the value of $\alpha$ needs to be on the order of $\frac{1}{O(m^2)}$. We can see that all conditions on $\alpha$ in terms of $m$ are satisfied if 
\begin{equation*}
    \alpha\le\frac{1}{4m^2}.
\end{equation*}
Set $\alpha = \frac{1}{4m^2}$ and then consider those conditions on $\gamma$ in terms of $m$ and $\alpha$. Note that \eqref{condition:2gamma1} is the tightest condition. Set $\gamma = \frac{1}{2}-\frac{1}{4^{20}m^2}$ to satisfy \eqref{condition:2gamma1}. Then, substitute the values of $\alpha$ and $\gamma$, we can verify that all conditions on $\gamma$ in terms of $m$ and $\alpha$ are satisfied, which proves our main theorem:
\begin{theorem}[Precise version of \Cref{thm:ratio_m>3}]
    By setting $\alpha = \frac{1}{4m^2}$, $\lambda = 4^{25/6}$, the \GSLEEPY algorithm with dynamic locking achieves a competitive ratio of $1.5-\frac{1}{4^{20}m^2}$.
\end{theorem}



%% file: alphagamma_m=3.tex
\section{Further Discussion on $m=3$}
\label{sec:m=3}

Here, we fix $\lambda = 1$. There is an interesting fact about this special case.

\begin{lemma}
\label{claim:nispushed}
    If $3\gamma-5\alpha-1 > 0$, the size of the last locking chain $|S|=1$.
\end{lemma}

\begin{proof}
    If the size of the last locking chain $|S|\ge2$, by the definition of the last locking chain and \Cref{claim:lockingchain_basic}, 
    $\hat{P}_\sigma(0,s_n)-P_\sigma(0,s_n)\ge(1-\alpha)p_n$. Then by \Cref{claim:lowerbound_pn}, \[p_n > \frac{3\gamma-\frac{3}{4}-6\alpha+\hat{P}_\sigma(0,s_n)- P_\sigma(0,s_n)}{\frac{5}{4}-2\alpha}\geq\frac{3\gamma-\frac{3}{4}-6\alpha+(1-\alpha)p_n}{\frac{5}{4}-2\alpha}.\] Therefore, under the condition that
\begin{empheq}[box=\shadowbox*]{equation}
    \label{conm=3:1}
    3\gamma-5\alpha-1 > 0,
    \tag{{Condition 1 for m=3}}
\end{empheq}
    we have $p_n>\frac{3\gamma-6\alpha-\frac{3}{4}}{\frac{1}{4}-\alpha}>1$, which contradicts the assumption that $\pi(\mathcal{J})=1$.
\end{proof}
\Cref{claim:nispushed} implies that the job $n$ is a pushed job. Thus, there will be exactly $m$ processing jobs at time $s_{n}$. Identify them by indices $j_1,j_2,\dots,j_m$. Define $\mathcal{J}_{C,m=3}=\{n,j_1,\dots,j_m\}$ as the critical set of tasks in this case. These $m+1$ jobs in $\mathcal{J}_{C,m=3}$ satisfy the following straightforward property.

\begin{proposition} [Stronger Version of \Cref{lem:earlyandlate}]
    For each critical job $j$, we call it
\begin{itemize}
    \item \textbf{Early Critical Job}, if $s_{j} < r_{n}$, and we have $p_{j} > \gamma$;
    \item \textbf{Late Critical Job}, if $s_{j} \geq r_{n}$, and we have $p_{j} \geq p_n$. 
\end{itemize}
\end{proposition}
\begin{proof}
    It is trivial that $p_n\ge p_n$. The processing time for early critical jobs exceeds $\gamma$, derived from $s_n-r_n>\gamma$, and late critical jobs' processing time is at least $p_n$ according to the LPT strategy.
\end{proof}

By the pigeonhole principle, at least $2$ of these $m+1$ jobs in $\mathcal{J}_{C,m=3}$ will be assigned to one machine in $\pi$. We denote this pair of jobs as jobs $a$ and $b$. Firstly, we eliminate the case that jobs $a$ and $b$ are both early critical jobs. Remark that we do not require the dynamic locking factor $\lambda$, because we do not have locked jobs. 

\begin{proposition}
\label{claim:m=3sasbrn}
    It is impossible that jobs $a$ and $b$ are both early critical jobs.
\end{proposition}

\begin{proof}
    Assume $s_a<s_b$. By our locking strategy and \Cref{lem:last_delay}, $s_b\geq s_a+\alpha p_a$ and $(1-\alpha)p_a>\gamma$. Then $p_a+p_b\ge 2(s_n-s_b)+s_b-s_a\ge 2\gamma+\alpha p_a>\frac{2-\alpha}{1-\alpha}\gamma$. Conditioned on
\begin{empheq}[box=\shadowbox*]{equation}
    \label{conm=3:2}
    \frac{2-\alpha}{1-\alpha}\gamma\ge1,
    \tag{Condition 2 for m=3}
\end{empheq}
$p_a+p_b>1$. This contradicts the supposition that $\pi(\mathcal{J})=1$.
\end{proof}

From \Cref{claim:m=3sasbrn}, we know that either $s_a\ge r_{n_k}$ or $s_b\ge r_{n_k}$. By \Cref{lem:idle_basic}, $\theta^+ \leq r_{n}$. Similarly, as in the general analysis, we will discuss the following two cases of this pair of jobs. 
\begin{enumerate}
    \item $\min\{s^\pi_a,s^\pi_b\} < \theta^+$.
    \item $\min\{s^\pi_a,s^\pi_b\} \geq \theta^+$.
\end{enumerate}

For the first case, assume $s_a^\pi<s_b^\pi$ without loss of generality. By \Cref{claim:pre_case1}, $s_a<r_n$ and $p_a>\gamma$. Then $s_a<r_n<s_b$. Similarly to the proof of \Cref{claim:sasbrnk_case1}, we get $s_a-r_a\ge s_a-s_a^\pi>\gamma$. We know that there won't be idle machines during $[r_a,s_a)$ and $[r_n,s_n)$. Since $C_a\geq s_n$, there is at least one busy machine during $(r_a, C_n)$. Assume the minimal release time is $r_{min}$. Therefore, $3(1-r_{min})\ge P_\pi=P_\sigma\ge s_n+p_n-r_{min}+2\gamma(\frac{3}{1+2\alpha}-1)>1+\gamma+2\gamma(\frac{3}{1+2\alpha}-1)-r_{min}$. Conditioned on 
\begin{empheq}[box=\shadowbox*]{equation}
    \label{conm=3:3}
    \gamma(\frac{6}{1+2\alpha}-1)\ge2,
    \tag{Condition 3 for m=3}
\end{empheq}
$3>1+\gamma(\frac{6}{1+2\alpha}-1)+2r_{min}\ge3$, which is a contradiction. The illustration of the impossibility of the second case is the same as \Cref{sec:case2}. Therefore, we should follow Conditions 9-15. 

In conclusion, to prove \Cref{thm:ratio_m=3}, we need not discuss the property of the last locking chain because we prove $|S|=1$ by a new condition of \eqref{conm=3:1}. Therefore, Conditions 1-3 can be removed. On the other hand, it also becomes easier to handle early critical pairs. The proof of \Cref{claim:m=3sasbrn} replaces Conditions 4-6 by \eqref{conm=3:2}. Then, for late critical pairs, in the first case, we also replace Conditions 7-8 with a weaker one-\eqref{conm=3:3}. In the second case, we keep the same proof as the general case. Therefore, we keep Conditions 9-15. 
\subsection{Conditions Required in This Special Case}
\begin{align}
    3\gamma-5\alpha-1 &> 0
    \tag{Condition 1 for m=3}\\
    \frac{2-\alpha}{1-\alpha}\gamma&\ge1
    \tag{Condition 2 for m=3}\\
    \gamma(\frac{6}{1+2\alpha}-1)&\geq2
    \tag{Condition 3 for m=3}\\
    1-2\alpha&\ge 0     
    \tag{Condition 9}\\
    \frac{5}{4}-6\alpha&\ge0   
    \tag{Condition 10}\\
    1-5\alpha+12\alpha^2&\ge 0\ \text{(always holds)}     
    \tag{Condition 11}\\
    3(1-2\gamma)+(12\alpha+1)\gamma-1&\le 0    
    \tag{Condition 12}\\
    \frac{1}{2}-6\alpha&\ge0     
    \tag{Condition 13}\\
    1-3\alpha(1-2\alpha)&\ge 0\ \text{(always holds)}    
    \tag{Condition 14}\\
    3(\frac{1}{2}-\gamma)+6\alpha (\frac{1}{2}+\gamma)-\frac{1}{2}&\le 0    
    \tag{Condition 15}
\end{align}

First, it is straightforward to verify Conditions 9, 10, 11, 13, 14 hold if $\alpha\le\frac{1}{12}$.
For the remaining conditions, it can be checked that $\alpha=0.07066$ and $\gamma = 0.4817$ are feasible for all of them. We plot a figure for help verifying. 

\begin{figure}[H]
    \centering
    \includegraphics[width=0.6\textwidth]{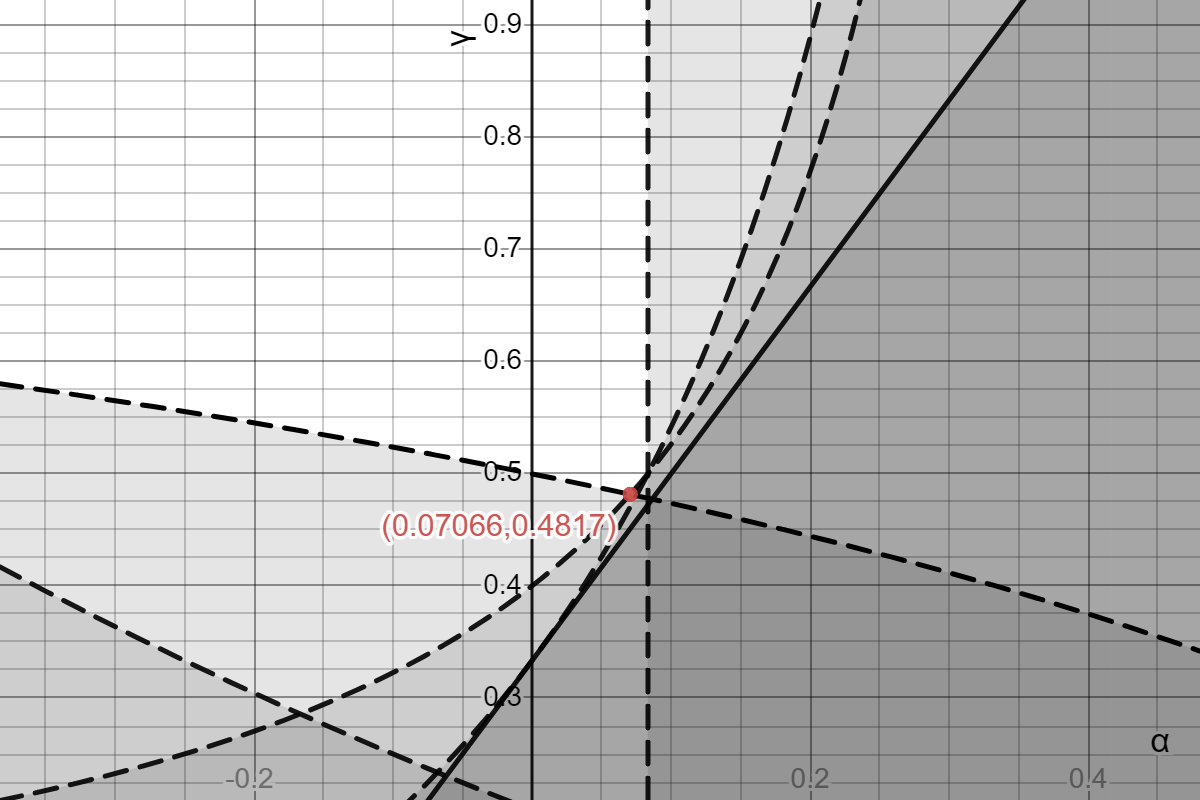}
    \caption{We plot all conditions in the figure and mark the infeasible region in grey. Finally, the white region means the feasible region where all constraints are satisfied.}
    \label{fig:enter-label}
\end{figure}

%% file: hard_instance.tex
\section{Hard Instance}
\label{sec:hard}

 Here, we prove the following theorem to demonstrate the necessity of introducing dynamic locking.
We recall \Cref{thm:dynamic locking}:

\dynamiclocking*

The proof is given by case-by-case analysis on the locking parameter $\alpha$ in the \GSLEEPY.

\subsection{Case 1: $\alpha \geq \frac{1}{2(m-1)}$}

Consider the following instance: $m$ jobs are released at time $0$, with processing time $1$. In the optimum schedule, by starting the $m$ jobs at time 0, we obtain the optimum makespan of $1$. In the \GSLEEPY algorithm's schedule, the last started job of the $m$ jobs would be $(m-1)\alpha$, so the completion time of the job is $(m-1)\alpha + 1 \geq 1.5$. In this case, the competitive ratio of \GSLEEPY is at least $1.5$. Refer to Figure \ref{fig:hard_instance_1}.

\begin{figure}[htbp]
    \centering
    \begin{subfigure}[b]{0.2\textwidth}
         \centering
         \includegraphics[width=\textwidth]{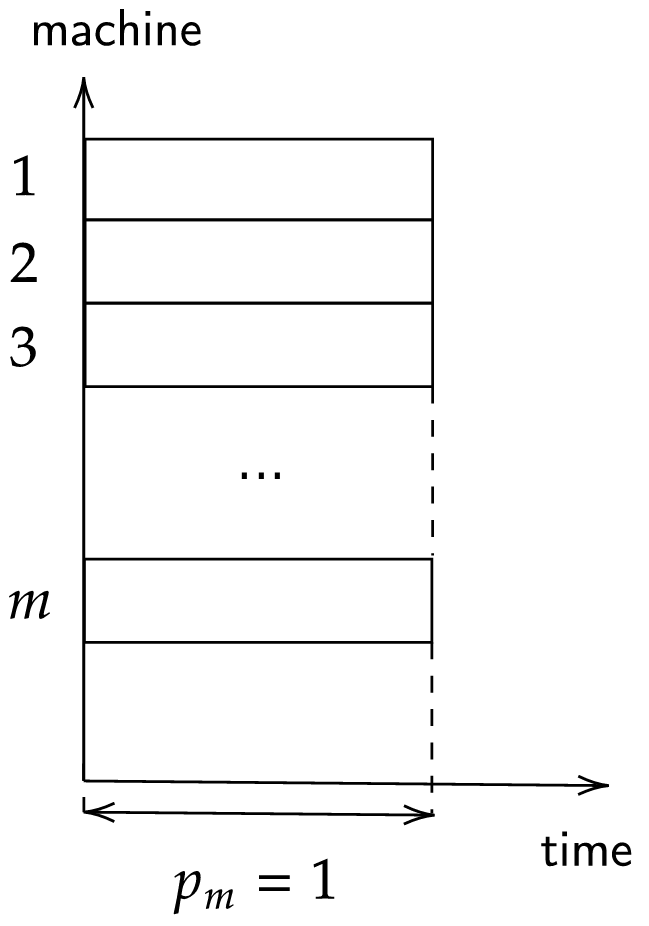}
     \end{subfigure}
     \begin{subfigure}[b]{0.32\textwidth}
         \centering
         \includegraphics[width=\textwidth]{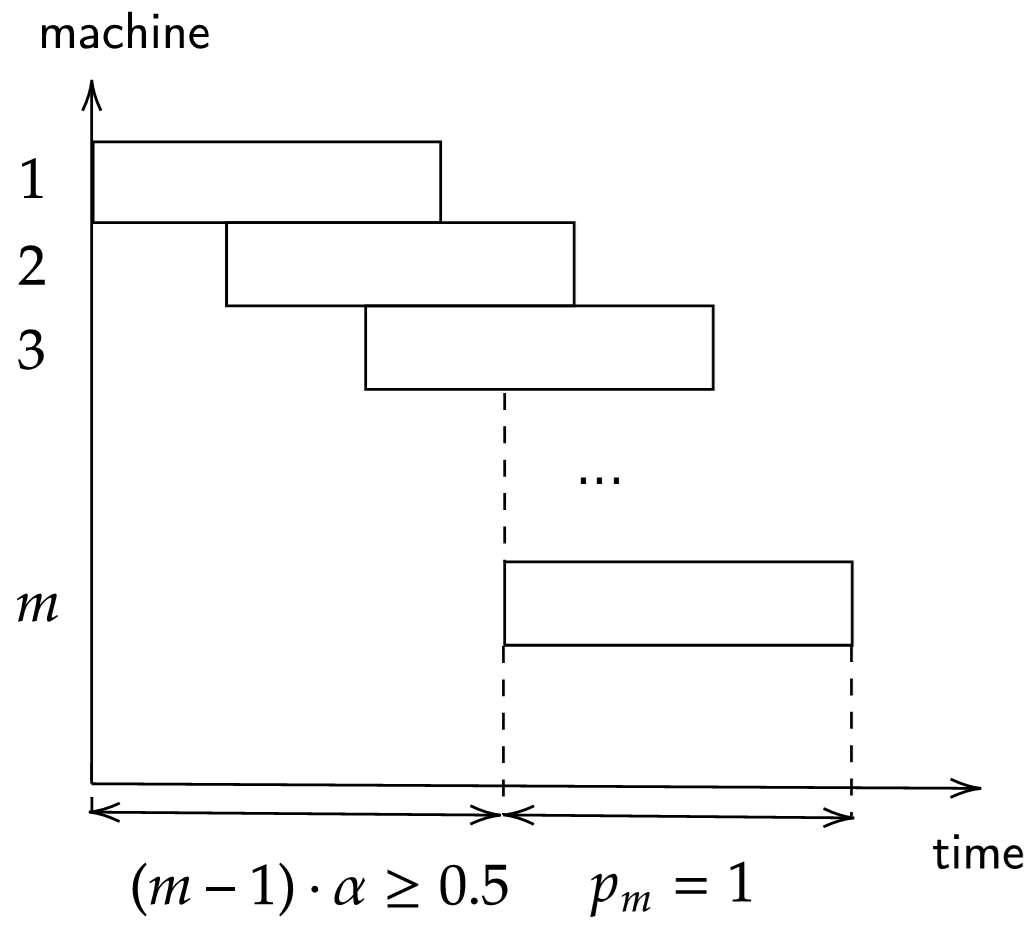}
     \end{subfigure}
    \caption{The schedule of $\pi$ (left) and  $\sigma$ (right) in the hard instance of \textbf{Case 1}.}
    \label{fig:hard_instance_1}
\end{figure}

\subsection{Case 2: $\alpha < \frac{1}{2(m-1)}$ ($m\geq6$)}

Consider the following instance:
\begin{itemize}
    \item $6$ jobs labeled from $1$ to $6$, with $\forall i \in [1, 6].~ r_i = 0$, $p_1 = \frac{1}{1 + (1-\alpha)^5}$, and $\forall i \in [2, 6].~ p_i = (1-\alpha)^{i-1} p_1$.
    \item $m - 3$ jobs labeled from $7$ to $m + 3$, with $\forall i\in [7,m+3].~ r_i = (1 - (1-\alpha)^5)p_1+\varepsilon$, where $\varepsilon \rightarrow 0$, and $p_i = 1 -  (1 - (1-\alpha)^5)p_1$. 
\end{itemize}

\subsubsection{\OPT's Schedule.}

The $i$-th machine processes the $i$-th and the $7-i$-th jobs in the first three machines and the completion time is $p_i + p_{7-i} = p_1 \cdot ((1-\alpha)^{i-1} + (1-\alpha)^{6-i}) \leq p_1 \cdot (1 + (1-\alpha)^{5}) = 1$. In the last $m-3$ machines, each of them processes one job labeled from $[7, m+3]$, the finish time is $ (1 - (1-\alpha)^{5})p_1+\varepsilon  + (1 -  (1 - (1-\alpha)^{5})p_1)= 1+\varepsilon \rightarrow 1$. Refer to \Cref{fig:hard_instance_case_2_pi} below. 
\begin{figure}[H]
    \centering
    \includegraphics[width=0.55\textwidth]{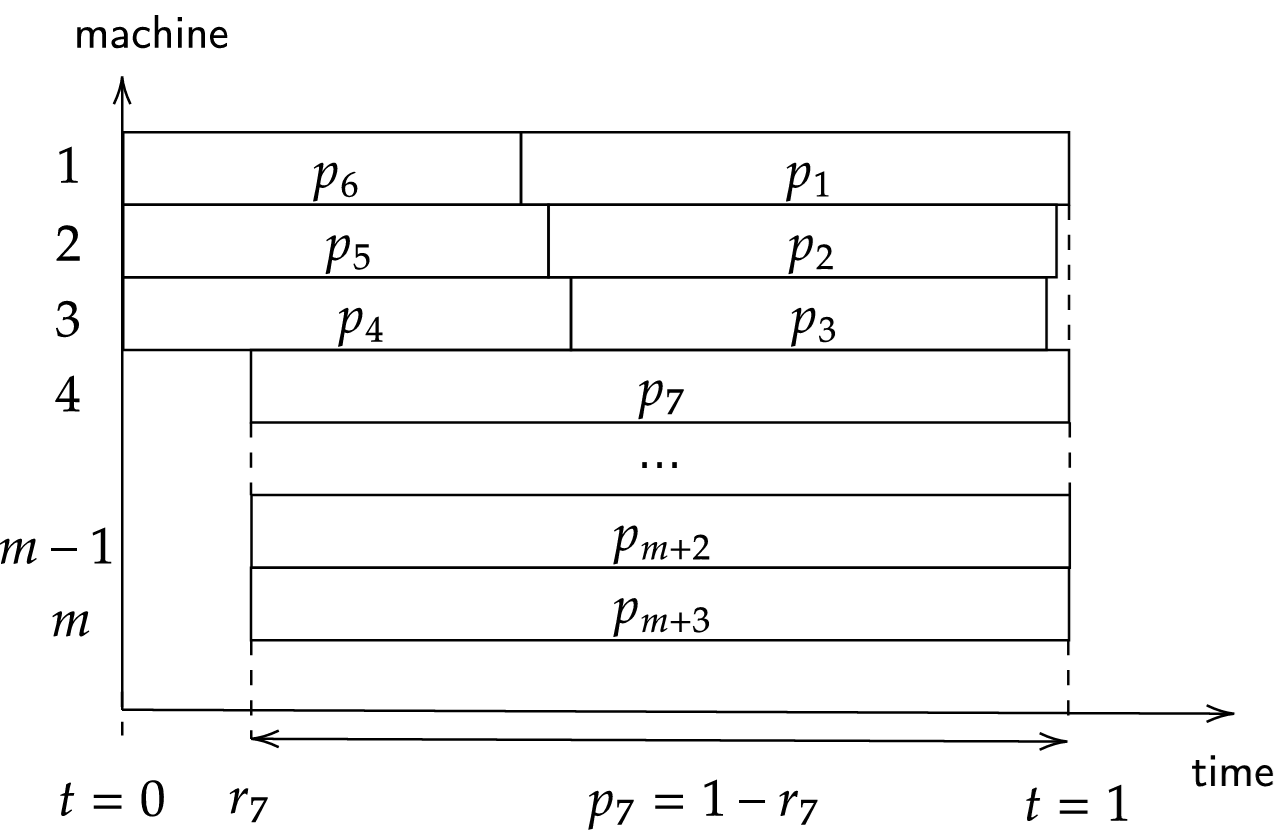}
    \caption{The schedule of $\pi$ in the hard instance of \textbf{Case 2}.}
    \label{fig:hard_instance_case_2_pi}
\end{figure}

\subsubsection{Generalized SLEEPY's Schedule.} 

For $1 \leq i \leq 6$, because the processing time decreases from $1$ to $6$, the jobs would be placed in this order and processed in different machines. 
Formally, we have $s_1=0$ and for all $2 \leq i \leq 6$, $s_i = \sum_{j = 1}^{i-1} \alpha p_j$. Therefore, for all $1 \leq i \leq 6$,
\begin{align*}
    C_i = s_i + p_i 
   &= \sum_{j = 1}^{i-1} \alpha p_j + (1-\alpha)^{i-1} p_1\\
   &= \sum_{j = 1}^{i-2}\alpha p_j + \alpha  (1 - \alpha)^{i-2}p_1 +  (1 - \alpha)^{i-1}p_1 \\
   &= \sum_{j = 1}^{i-2}\alpha p_j +  (1 - \alpha)^{i-2}p_1 \\
   &= \dots = p_1.
\end{align*}
Note that $s_{6} = p_1 - p_{6}  = p_1 (1 - (1-\alpha)^{5}) = r_{7}-\varepsilon<r_7$, the last $m - 3$ jobs are released just after the start of the job $6$. Because only $m-6$ machines are free before $p_1$, three of the last $m-6$ jobs must be scheduled after $p_1$ or even after $r_7+p_7=1 > p_1$. Without loss of generality, we call them $m+1$, $m+2$, and $m+3$. By the definition of \GSLEEPY, the start time of the last job $s_{m+3}$ must be large, where
\[
    s_{m+3} = s_{m+1}+ 2\alpha (1- (1 - (1-\alpha)^{5})p_1) = p_1 + 2\alpha (1- (1 - (1-\alpha)^{5})p_1). 
\]
Therefore, 
\begin{align*}
    C_{m+3} &= p_1 + 2\alpha (1- (1 - (1-\alpha)^{5})p_1) + p_{m+3}\\
    &= \frac{1}{1 + (1-\alpha)^5} + (1+2\alpha) (1-\frac{1 - (1-\alpha)^{5}}{1 + (1-\alpha)^5})\\
    &= p_1+(1+2\alpha)p_7.
\end{align*}
Refer to \Cref{fig:hard_instance_case_2_sigma} below. 

\begin{figure}[H]
    \centering
    \includegraphics[width=0.8\textwidth]{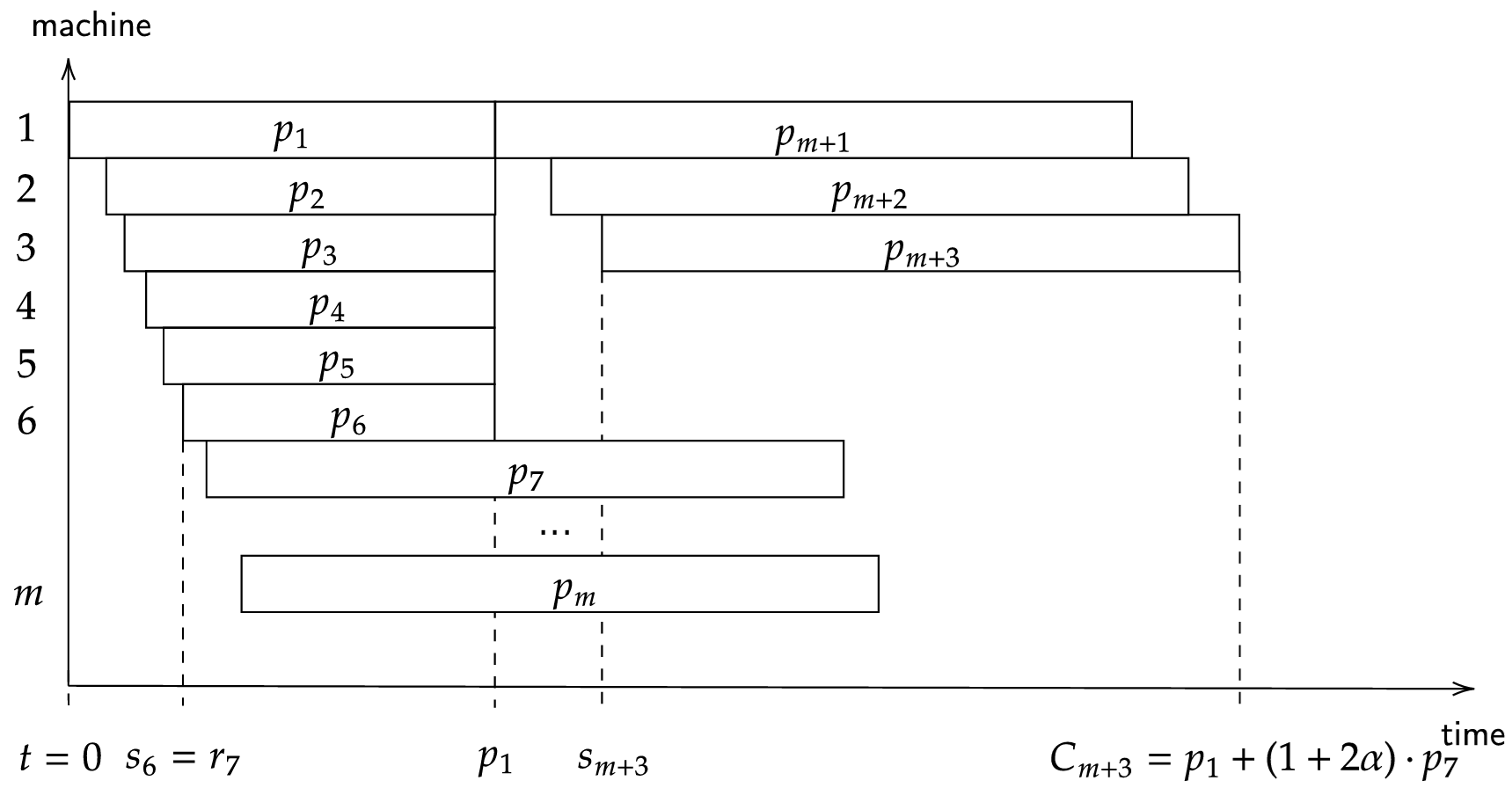}
    \caption{The schedule of $\sigma$ in the hard instance of \textbf{Case 2}.}
    \label{fig:hard_instance_case_2_sigma}
\end{figure}
Let $f(\alpha) = \frac{1}{1 + (1-\alpha)^5} + (1+2\alpha) (1-\frac{1 - (1-\alpha)^{5}}{1 + (1-\alpha)^5})$. We observe that for all $\alpha \leq \frac{1}{2(m-1)}$, $f(\alpha)\geq 1.5$. We plot the function in the following \Cref{fig:hardnessratio} for verification. The competitive ratio of \GSLEEPY is at least $1.5$ in this case.





\begin{figure}[H]
    \centering
    \includegraphics[width=0.5\textwidth]{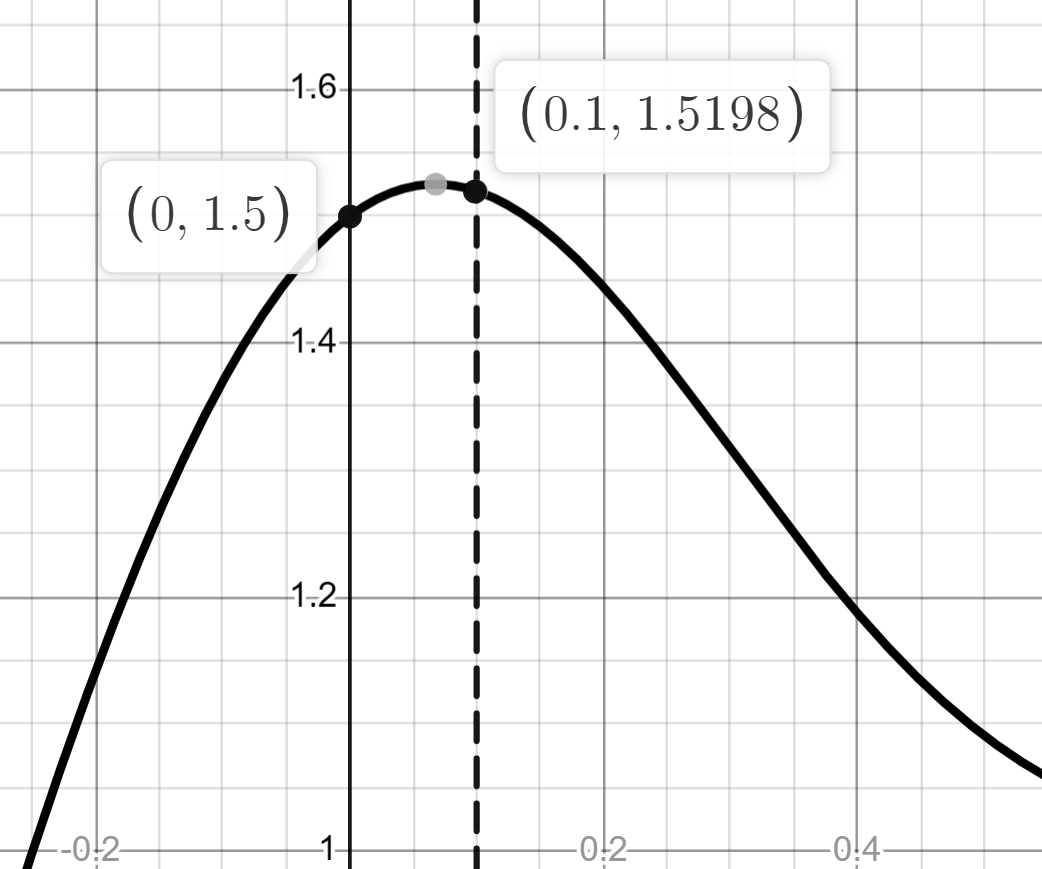}
    \caption{Function $f(\alpha)$. We have $\forall \alpha \leq \frac{1}{2(m-1)},~f(\alpha)\geq 1.5$. Note that the vertical line is $\alpha = 0.1 = \frac{1}{2(m-1)}$ when $m=6$.}
    \label{fig:hardnessratio}
\end{figure}

%% file: conclusion.tex
\section{Conclusion}
In conclusion, our main contribution is to show that we can beat the $1.5$-competitive ratio of \LPT for every constant $m$, marking an important step toward understanding whether \LPT is the optimal algorithm for this over-time online makespan minimization problem. 
We believe that our results and techniques can inspire further work toward fully resolving this open question:
\begin{itemize}
    \item If we believe that \LPT is indeed optimal when $m$ becomes large, then it would be necessary to construct a counterexample involving an infinite number of machines. This is because, for any constant number of machines, our algorithm achieves a competitive ratio strictly below $1.5$. Prior to our results, many researchers may have believed that a counterexample with some constant $m$, for example $m=4$, could already demonstrate the tight bound of $1.5$.
    \item On the other hand, if we believe that there exists a better algorithm that can beat $1.5$ for general $m$, then several challenges must be addressed. In our paper, we use a $o(1)$ locking parameter; otherwise, the algorithm could waste too much processing power, making it difficult to use efficiency arguments. To improve beyond $1.5$ for general $m$, one would need to apply an $\Omega(1)$ locking parameter and resolve the associated efficiency issues by designing a more refined version of the left-over lemma. Furthermore, the complications arising from the last locking chain become even more severe, suggesting that a similar idea to dynamic locking must also be incorporated.
\end{itemize}

On the other hand, we also believe that our algorithmic idea is not merely a theoretical trick for analysis. In this over-time model, locking machines for a period to wait for more future information can play an important role in improving \LPT. Our dynamic locking rule suggests that a dynamic locking strategy may outperform a static one. Although in practice the locking parameter and the locking strategy may not be set exactly as in our work, it is an interesting direction to investigate whether there exists a suitable way to lock machines that can enhance the performance of \LPT in practical settings.